\documentclass[3p]{elsarticle}
\usepackage{pstricks}
\usepackage{pst-eps}
\usepackage{natbib}

\usepackage{amsthm,amsmath,mathrsfs,amssymb}
\usepackage{pifont}
\usepackage{epstopdf}
\usepackage{bbm}
\usepackage{color}
\usepackage{comment}
\usepackage{multirow}

\usepackage{float}
\usepackage{stfloats}
\usepackage[tight]{subfigure}
\usepackage{graphicx,color,url}

\usepackage{keyval}
\usepackage{algorithm}
\usepackage{algorithmic}
\usepackage{tabularx,booktabs}
\usepackage[colorlinks, citecolor=blue]{hyperref}
\usepackage{makecell}
\usepackage{stackengine}
\usepackage{bm}

\newsavebox{\tempbox}

\newtheorem{theorem}{Theorem}

\newtheorem{remark}{Remark}
\newtheorem{lemma}{Lemma}

\begin{document}

\begin{frontmatter}

\title{Volumetric Spline Parameterization for Isogeometric Analysis}

\author{Maodong Pan}

\author{Falai Chen\corref{cor}}
\ead{chenfl@ustc.edu.cn}
\cortext[cor]{Corresponding author.}

\author{Weihua Tong}

\address{School of Mathematical Sciences, University of Science and Technology of China, Hefei, Anhui, 230026, PR China}

\begin{abstract}
Given the spline representation of the boundary of a three dimensional domain, constructing a volumetric spline parameterization of the domain (i.e., a map from a unit cube to the domain) with the given boundary is a fundamental problem in isogeometric analysis. A good domain parameterization should satisfy the following criteria: (1) the parameterization is a bijective map; and (2) the map has lowest possible distortion.  However, none of the state-of-the-art volumetric parameterization methods has fully addressed the above issues. In this paper, we propose a three-stage approach for constructing volumetric parameterization satisfying the above criteria. Firstly, a harmonic map is computed between a unit cube and the computational domain. Then a bijective map modeled by a max-min optimization problem is computed in a coarse-to-fine way, and an algorithm based on divide and conquer strategy is proposed to solve the optimization problem efficiently. Finally, to ensure high quality of the parameterization, the MIPS (Most Isometric Parameterizations) method is adopted to reduce the conformal distortion of the bijective map. We provide several examples to demonstrate the feasibility of our approach and to compare our approach with some state-of-the-art methods. The results show that our algorithm produces bijective parameterization with high quality even for complex domains.
\end{abstract}


\end{frontmatter}

\section{Introduction}
\label{sec:intro}
Isogeometric analysis (IGA) has recently become a hotspot in numerical analysis and geometric modeling communities since
it integrates two related disciplines: Computer Aided Design (CAD) and Computer Aided Engineering (CAE)~\citep{hughes2005isogeometric}.
IGA overcomes the unnecessary data exchange between CAD models and the simulation software, and it has been successfully applied in various disciplines, such as structural vibration~\citep{cottrell2006isogeometric}, shell analysis~\citep{benson2011large}, phase transition phenomena~\citep{gomez2008isogeometric} and shape optimization~\citep{qian2010full,nguyen2012isogeometric}.

However, CAD systems provide only the boundary representation of an object (computational domain), while CAE typically requires the interior parameterization of the object. Thus constructing a volume parametric representation for the computational domain from the given boundary data is an essential step in IGA, and it is called $\emph{volumetric parameterization}$ (see Fig.~\ref{DomainParameterization}). Volumetric parameterization has great effect on the accuracy and efficiency in subsequent analysis~\citep{cohen2010analysis,xu2011parameterization,pilgerstorfer2014bounding,xu2015two}.
As sated in~\citep{falini2015planar,xu2018constructing,pan2018low}, a good volumetric parameterization should meet two basic requirements: firstly, it doesn't have self-intersections, i.e., the map from the parametric domain (generally a unit cube) to the computational domain is injective; secondly, the distortion of the map should be as small as possible, i.e., the volumes and angles after mapping should be preserved as much as possible. So far several approaches have been proposed to solve the volumetric parameterization problem, and they can be generally classified into two categories: (1) volumetric spline parameterization from boundary triangulations~\citep{martin2009volumetric,escobar2011new,zhang2012solid,wang2013trivariate,liu2014volumetric}; and (2) analysis-suitable volumetric parameterization from spline boundaries~\citep{aigner2009swept,nguyen2010parameterization,zhang2013conformal,xu2013constructing,xu2013analysis,xu2013optimal,xu2014high,wang2014optimization,lopez2017spline}. However, none of these methods have fully addressed the above issues. For example, no method guarantees the bijectivity of the parameterization, especially for complex computational domains. Thus constructing high-quality volumetric parameterizations of complex computational domains remains a big challenge in IGA.

\begin{figure}[!htbp]
    \centering
    \subfigure[Boundary spline surfaces]{
        \label{DomainParameterization:BoundaryNurbs}
        \includegraphics[width=0.25\textwidth]{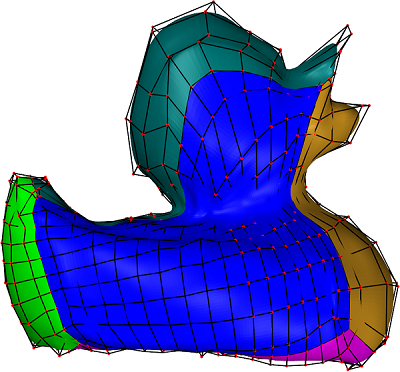}
    }
    \hspace{0.07\textwidth}
    \subfigure[A map from the parametric domain to the computational domain]{
        \label{DomainParameterization:VolumetricParameterization}
        \includegraphics[width=0.60\textwidth]{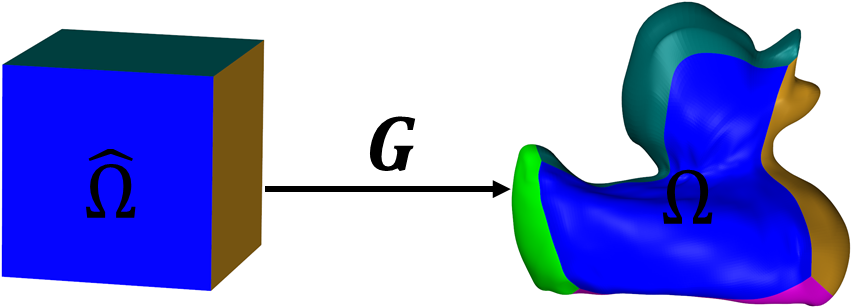}
    }
	\caption{Volumetric parameterization:~\subref{DomainParameterization:BoundaryNurbs} the six input B-spline boundary surfaces,~\subref{DomainParameterization:VolumetricParameterization} a map $\mathbf{G}$ from the unit cube $\hat{\Omega}$ to the computational domain $\Omega$.}
    \label{DomainParameterization}
\end{figure}

This paper describes a new approach for constructing high-quality volumetric parameterization. Specifically, suppose we are given the six spline boundary surfaces of a computational domain (of genus-zero), our goal is to construct a trivariate spline representation for the computational domain such that the parameterization is bijective and has low distortion. We propose a three-stage approach to solve the problem. Firstly, an initial parameterization is obtained by computing a harmonic map from the parametric domain (a unit cube) to the computational domain. Then we compute a bijective parameterization of the computational domain by solving a max-min optimization problem in a coarse-to-fine way. Finally, we improve the parameterization quality by minimizing the conformal distortion of the map using the MIPS method.

The remainder of this paper is organized as follows. Section~\ref{sec:work} reviews some related work about domain parameterization. Section~\ref{sec:preliminaries} presents some preliminary knowledge about the representation of volumetric parameterization, the sufficient condition for an injective parameterization and the distortion measurement of a parametrization which will be used in our method. In Section~\ref{sec:approach}, we propose three mathematical models followed by algorithms to compute a high-quality map for volumetric parameterization. Some examples are demonstrated in Section~\ref{sec:result} to show the effectiveness of the proposed method, and comparisons with the nonlinear optimization methods~\citep{xu2013analysis,wang2014optimization} are also provided. Finally, we conclude the paper with a summary and future work in Section~\ref{sec:conclusion}.

\section{Related work}
\label{sec:work}

In this section, we will review some related works on domain parameterization and emphasis will be put on volumetric parameterization.

For planar domain parameterization, a direct solution is the discrete Coons patches introduced by Farin and Hansford~\citep{farin1999discrete}. Gravesen et al.~\citep{gravesen2012planar} put forward a spring model to solve the problem. These two methods are simple but the resulting parameterization may not be injective. To ensure a bijective parameterization, several approaches were proposed by solving nonlinear optimization problems~\citep{xu2011parameterization,nian2016planar,ugalde2018injectivity}. For complex computational domains, however, single patch representations do not provide sufficient flexibility, and multi-patch structures were put forward to tackle the parameterization problem~\citep{xu2015two,buchegger2017planar,xu2018constructing,xiao2018computing}. Recently, Juettler et al.~\citep{juttler2017low} and Pan et al.~\citep{pan2018low} investigated the low-rank parameterization of planar domains which can improve the efficiency of assembling stiffness matrices in IGA.

Compared with planar case, volumetric parameterization is more challenging. Depending on the boundary representation (triangular meshes or splines) of a 3D computational domain, volumetric parameterizaion methods can be classified into two categories. For a domain whose boundary is represented by triangular meshes, Martin et al.~\citep{martin2009volumetric} proposed a trivariate B-spline parameterization method based on discrete volumetric harmonic functions. In~\citep{escobar2011new}, the authors focused on constructing a solid T-spline from the parametric mapping of tetrahedral meshes generated by meccano method from the input surface mesh. Zhang et al.~\citep{zhang2012solid} developed a procedure to construct a T-spline representation of a genus-zero solid.  When the computational domain has a topology of non-zero genus, polycube is a useful tool to solve the parameterization problem~\citep{li2007harmonic,wang2008polycube,wang2008user}. To extend the work~\citep{zhang2012solid} to domains with arbitrary topology, Wang et al.~\citep{wang2013trivariate} used the polycube mapping combined with subdivision and pillowing techniques to generate high-quality T-spline representations. Inspired by CSG Boolean operations, Liu et al.~\citep{liu2014volumetric} firstly built the polycubes from the input triangle mesh based on Boolean operations and then computed the volumetric parameterization by performing an octree subdivision on the polycubes. Later on skeleton-based polycube constructions were used to generate T-spline parameterizations~\citep{liu2015feature}. For computational domains with spline boundary representations,  Aigner et al.~\citep{aigner2009swept} presented a method for generating NURBS parameterizations of swept volumes via sweeping a closed curve. Nguyen et al.~\citep{nguyen2010parameterization} applied a sequence of harmonic maps to construct a volumetric parameterization, and the harmonic maps were obtained by solving some variational problems. In~\citep{zhang2013conformal}, the construction of volumetric conformal T-spline parameterization from boundary T-spline representation was studied by using octree structure and boundary offset. An optimization-based method was proposed in~\citep{xu2013analysis}, where the B-spline parameterization is found by solving a constraint optimization problem. A similar technique was proposed by the same authors in~\citep{xu2013constructing}. From six boundary B-spline surfaces, Wang and Qian~\citep{wang2014optimization} proposed an efficient method by combining the constraint aggregation and hierarchical optimization technique to obtain valid trivariate B-spline solids. Recently, a method based on a new T-mesh untangling and smoothing procedure was proposed to construct T-spline parameterizations for 2D and 3D geometries~\citep{lopez2017spline}.

Despite various methods for volumetric parameterization were proposed, however to the best of our knowledge, none of the methods can ensure the bijectivity of the parameterization, and the quality of the parameterization is not well investigated. The goal of the current paper is to present a new volumetric parameterization technique which can guarantee the bijectivity of the parameterization and has lowest possible distortion.

\section{Preliminaries}
\label{sec:preliminaries}
In this section, we present some preliminary knowledge about the representation of volumetric parameterizations, the sufficient condition for a bijective parameterization and the distortion measurement of a parameterization.

\subsection{Representation of volumetric parameterization}
In this paper, we assume a simply connected computational domain $\Omega$ is parameterized by a vector-valued trivariate tensor product B-spline function:
\begin{equation}\label{B-spline rep}
\begin{aligned}
{\mathbf{G}}(\xi,\eta, \zeta)&:=(u(\xi,\eta,\zeta),v(\xi,\eta,\zeta),w(\xi,\eta,\zeta))\\
&:=\sum_{i=0}^{l}\sum_{j=0}^{m}\sum_{k=0}^{n}\emph{\textbf{P}}_{ijk}N_{i}^{p}(\xi)N_{j}^{q}(\eta)N_{k}^{r}(\zeta),\quad (\xi,\eta,\zeta)\in \hat\Omega:=[0,1]^3,
\end{aligned}
\end{equation}
where $\emph{\textbf{P}}_{ijk}=\{x_{ijk},y_{ijk},z_{ijk}\}\in \mathbb{R}^3$ are the control points, $N_{i}^{p}(\xi)$, $N_{j}^{q}(\eta)$ and $N_{k}^{r}(\zeta)$ are the B-spline basis functions of degree $p$, $q$ and $r$ w.r.t the knot sequences $U$, $V$ and $W$ in $[0,1]$ respectively. $U\times V \times W$ defines a tensor product mesh $\mathcal T$ over which the B-spline function $\mathbf G$ is a piecewise polynomial, i.e, $\mathbf G$ is a polynomial over each cell (a cube) of the mesh $\mathcal T$.

\subsection{Sufficient condition for a bijective parameterization}
\label{sec:BijectiveCondition}

To guarantee the bijectivity of a map, a necessary condition is the positivity of the Jacobian of the map. In general, we have

\begin{lemma}\label{injectivecondition}(\citep{more1973p})
A continuously differential map $\mathbf G$ is locally bijective provided its Jacobian, denoted as $\det(J_{\mathbf G})$, does not vanish on the parametric domain, and the global bijectivity of $\mathbf G$ is guaranteed if it is locally bijective
on the parametric domain, and the computational domain is simply connected and the restriction of $\mathbf G$ on the domain boundary is bijective.
\end{lemma}
In this work, we always assume that the computational domain $\Omega$ is simply connected, and a bijective boundary correspondence of the parametric domain $\hat{\Omega}$ and computational domain $\Omega$ is established. Thus in this case, a locally bijective parameterization is also globally bijective.

From Lemma~\ref{injectivecondition}, it is obvious that if the Jacobian $\det(J_{\mathbf G}(\xi,\eta,\zeta))$ is positive everywhere on the parametric domain, i.e.,
\begin{equation}
    \begin{split}
\det(J_{\mathbf G}(\xi,\eta,\zeta))=\left|\begin{array}{cccc}
    u_{\xi} & u_{\eta} & u_{\zeta} \\
    v_{\xi} & v_{\eta} & v_{\zeta}\\
    w_{\xi} & w_{\eta} & w_{\zeta}
\end{array}\right|>0, \quad \forall (\xi,\eta,\zeta)\in [0,1]^3,
    \end{split}
\end{equation}
then the parametrization $\mathbf G$ is globally bijective.

Owning to the good properties of B-splines~\citep{morken1991some}, the Jacobian of a volumetric B-spline parameterization (\ref{B-spline rep}) is itself a higher order trivariate B-spline which has the following form:
\begin{equation}\label{highorderBspline}
\det(J_{\mathbf G}(\xi,\eta,\zeta))=\sum_{i=0}^{3l-1}\sum_{j=0}^{3m-1}\sum_{k=0}^{3n-1}G_{ijk}N_{i}^{3p-1}(\xi)N_{j}^{3q-1}
(\eta)N_{k}^{3r-1}(\zeta), \quad \forall (\xi,\eta,\zeta)\in [0,1]^3,
\end{equation}
where $\{G_{ijk}\}_{i,j,k=0}^{3l-1,3m-1,3n-1}$ denote the coefficients of the higher order B-spline. From the convex hull property of B-spline function, a sufficient condition for the bijectivity of the  parameterization (\ref{B-spline rep}) is that all the coefficients $\{G_{ijk}\}$ are positive. However, this is a very strict condition for a bijective parameterization as reported in~\cite{xu2013optimal}. In order to relax the condition, we convert the B-spline representation (\ref{highorderBspline}) into B{\'e}zier forms:
\begin{equation}\label{highorderBezierSurface}
 B^\gamma(\xi,\eta,\zeta)=\sum_{i=0}^{3p-1}\sum_{j=0}^{3q-1}\sum_{k=0}^{3r-1}B_{ijk}^\gamma B_{i}^{3p-1}(\xi)B_{j}^{3q-1}(\eta)B_{k}^{3r-1}(\zeta), \quad (\xi,\eta,\zeta)\in {\hat\Omega}_{\gamma}, \, \gamma=1,2,\ldots, \tau,
\end{equation}
where $\{B^\gamma_{ijk}\}$ are the control coefficients, $B_{i}^{3p-1}(\xi)$, $B_{j}^{3q-1}(\eta)$ and $B_{k}^{3r-1}(\zeta)$ denote the Bernstein basis polynomials of degree $3p-1$, $3q-1$ and $3r-1$ respectively, and ${\hat\Omega}_{\gamma}$ is a cell of the tensor product mesh $\mathcal T$. The coefficients of these trivariate B{\'e}zier polynomials can be obtained by multiple knot insertion using the blossoming technique~\citep{farin2002curves}. Now if all the coefficients $\{B^\gamma_{ijk}\}$ are positive, then the Jacobian $\det(J_{\mathbf G})$ is also positive on the parametric domain $\hat\Omega$. On the other hand, if some coefficients in $B^\gamma(\xi,\eta,\zeta)$ are negative, we can further subdivide the B{\'e}zier polynomial $B^\gamma(\xi,\eta,\zeta)$
into eight sub-polynomials at the parametric values $(\xi, \eta, \zeta)=(0.5, 0.5, 0.5)$ and check the positivity of the coefficients of each sub-polynomial. Again if the coefficients of all the B{\'e}zier polynomials (including polynomials after subdivision) are positive, the Jacobian $\det(J_{\mathbf G})$ is also positive. This process can be repeated until certain level of subdivisions. In this way, we can obtain a set of relaxed sufficient conditions for the bijectivity of the map $\mathbf G$. However, a risk of such process is that the number of inequality constraints will be huge~\citep{xu2018constructing}.

Instead of using the above sufficient conditions as constraints, in this paper we will directly impose positivity constraints of the Jacobian at a set of collocation points, and apply the above sufficient conditions to check the positivity of the Jacobian on the whole parametric domain $\hat\Omega$.

\subsection{The distortion measurement of a parameterization}
\label{sec:distortiondefinition}
The distortion of a differentiable map $\mathbf{G}$ at a point $\mathbf{x}$ is some measure about how $\mathbf{G}$ changes at the vicinity of $\mathbf{x}$. Here we introduce two common measures of distortion: volume distortion and conformal (angular) distortion. In addition, the fairness measure for a map is also introduced.

\medskip

\noindent\textbf{Volume distortion.} The volume distortion is the determinant of the Jacobian of the map divided by the volume of $\Omega$:
\begin{equation}
D_{vol}(\mathbf{G};\mathbf{x})=\frac{\det(J_{\mathbf G}({\mathbf x}))}{\text{Vol}(\Omega)},
\end{equation}
Thus $D_{vol}(\mathbf{G};\mathbf{x})$ is also called the scaled Jacobian of the map $G$. The optimal map is
the one that satisfies $D_{vol}(\mathbf{G};\mathbf{x})=1$ everywhere.

\medskip

\noindent\textbf{Conformal distortion.} A conformal map preserves angles and thus it can produce iso-parametric orthogonal nets. Denote the singular values of $J_{\mathbf{G}}(\mathbf{x})$ as $\sigma_1(\mathbf{G},\mathbf{x})$, $\sigma_2(\mathbf{G},\mathbf{x})$ and $\sigma_3(\mathbf{G},\mathbf{x})$, then the conformal distortion can be defined by MIPS function~\citep{hormann1999mips}:
\begin{equation}\label{disconformal}
\begin{aligned}
D_{con}(\mathbf{G};\mathbf{x})&= \frac{1}{8}\left(\frac{\sigma_1(\mathbf{x})}{\sigma_2(\mathbf{x})}+\frac{\sigma_2(\mathbf{x})}{\sigma_1(\mathbf{x})}\right)
\left(\frac{\sigma_2(\mathbf{x})}{\sigma_3(\mathbf{x})}+\frac{\sigma_3(\mathbf{x})}{\sigma_2(\mathbf{x})}\right)
\left(\frac{\sigma_3(\mathbf{x})}{\sigma_1(\mathbf{x})}+\frac{\sigma_1(\mathbf{x})}{\sigma_3(\mathbf{x})}\right)\\
&=\frac{1}{8}(\|J_{\mathbf{G}}(\mathbf{x})\|_F^2\|J_{\mathbf{G}}(\mathbf{x})^{-1}\|_F^2-1),
\end{aligned}
\end{equation}
where $||\cdot||$ stands for the Frobenius norm.
It's easy to see that $D_{con}(\mathbf{G};\mathbf{x})$ achieves the minimal value $1$ if and only if $\sigma_1(\mathbf{x})=\sigma_2(\mathbf{x})=\sigma_3(\mathbf{x})$, that is, $\mathbf{G}$ is a conformal map.

\medskip

\noindent\textbf{Fairness.} Besides that the map should have smallest distortion, we also hope the map is smooth and fair. This can be characterized by the triharmonic quantity:
\begin{equation}
D_{fair}(\mathbf{G};\mathbf{x})=\|Hu(\mathbf{x})\|_F^2+\|Hv(\mathbf{x})\|_F^2+\|Hw(\mathbf{x})\|_F^2,
\end{equation}
where $Hu$, $Hv$ and $Hw$ are the Hessians of $u$, $v$ and $w$ respectively.

\section{A three-stage parameterization method of computational domains}
\label{sec:approach}

Given the B-spline representations of the six boundary surfaces of a computational domain $\Omega$, our goal is to compute a trivariate B-spline representation for $\Omega$, i.e., a map from the unit cube $\hat{\Omega}=[0,1]^3$ to $\Omega$ which is bijective and has low distortion. In this section, we introduce a three-stage approach towards the goal.

\subsection{Computing a harmonic map}
\label{subsection:initialmap}
The first stage of our method is to compute a good initial parameterization which is critical for accelerating the convergence of subsequent optimization. There are several ways to construct an initial parameterization such as discrete Coons volumetric interpolation~\citep{farin1999discrete} and deformation based method~\citep{wang2014optimization}. In this paper, we compute a harmonic map as the initial parameterization. The computation is similar to those in~\citep{wang2003volumetric,su2017volume} except that our representation is spline-based.

According to the smooth harmonic map theory~\citep{schoen1997lectures}, a harmonic map $\mathbf{G}$ is a function satisfying Laplace's equation, i.e.,
\begin{equation*}
\Delta \mathbf{G}=\mathbf{0},
\end{equation*}
where $\Delta=\frac{\partial^2}{\partial \xi^2}+\frac{\partial^2}{\partial \eta^2}
+\frac{\partial^2}{\partial \zeta^2}$ is the Laplace operator.
Thus the harmonic map under the given boundary information can be found by minimizing the following energy
\begin{equation}\label{harmonicmap}
\begin{split}
\min_{\mathbf{G}}\quad &\int_{\hat{\Omega}}\|\Delta \mathbf{G}\|^2\mathrm{d}\xi\mathrm{d}\eta\mathrm{d}\zeta\\
\mathrm{s.t.}\quad &\mathbf{G}\big|_{\partial \hat{\Omega}}\ \text{is given}.
\end{split}
\end{equation}
This is a quadratic optimization problem and the solution can be obtained by solving a sparse and symmetric linear system of equations. The preconditioned conjugate gradient method with incomplete Cholesky factorization is applied in solving the linear system. From Fig.~\ref{HarmonicMapping}, we can see that the initial parameterization constructed using harmonic map is nearly valid (i.e., the parameterization is bijective except at regions near the boundary), which is superior to most of the other initialization methods.

\begin{figure}[!htbp]
    \centering
    \subfigure[]{
        \label{HarmonicMapping:BoundaryNurbs}
        \includegraphics[width=0.235\textwidth]{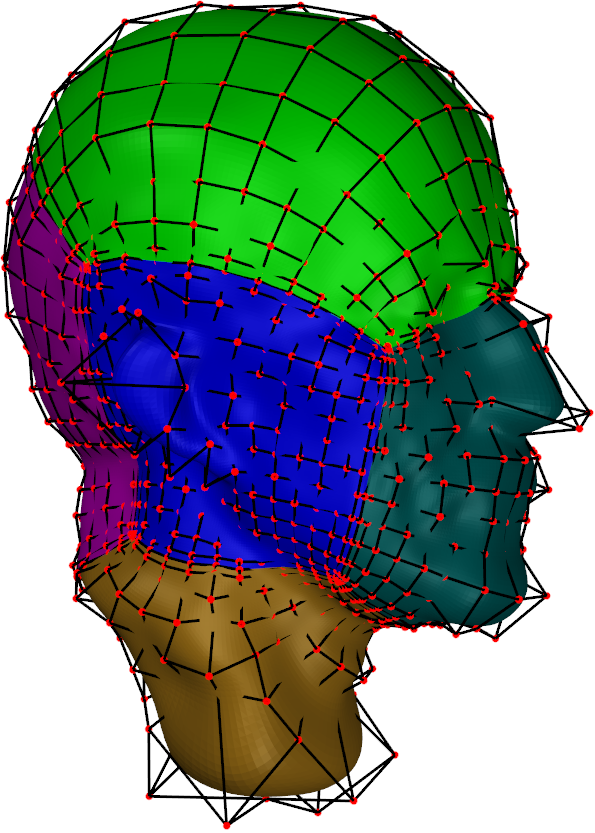}
    }
    \hspace{0.00\textwidth}
    \subfigure[]{
        \label{HarmonicMapping:InteriorOne}
        \includegraphics[width=0.24\textwidth]{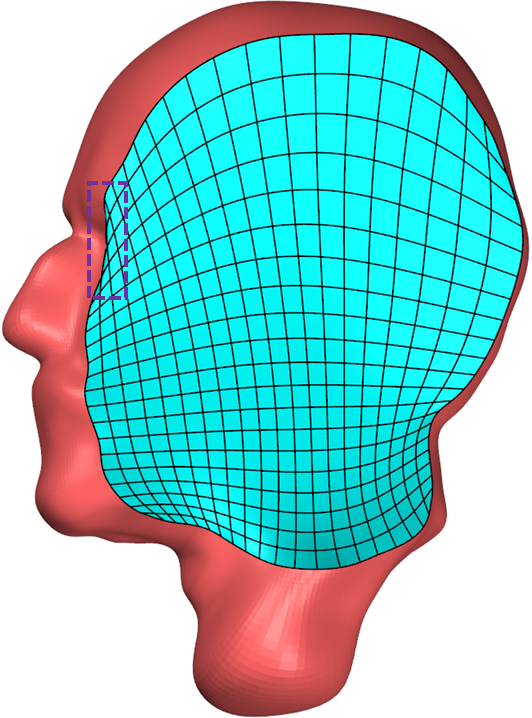}
    }
    \hspace{0.00\textwidth}
    \subfigure[]{
        \label{HarmonicMapping:InteriorTwo}
        \includegraphics[width=0.23\textwidth]{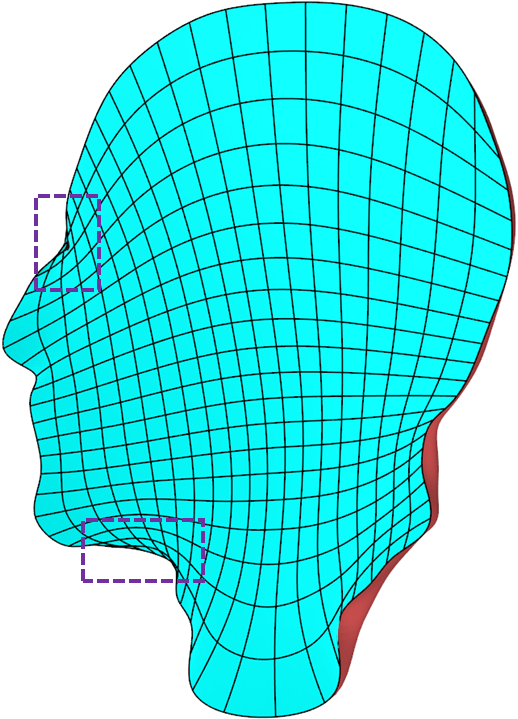}
    }
    \hspace{0.00\textwidth}
    \subfigure[]{
        \label{HarmonicMapping:InteriorThree}
        \includegraphics[width=0.165\textwidth]{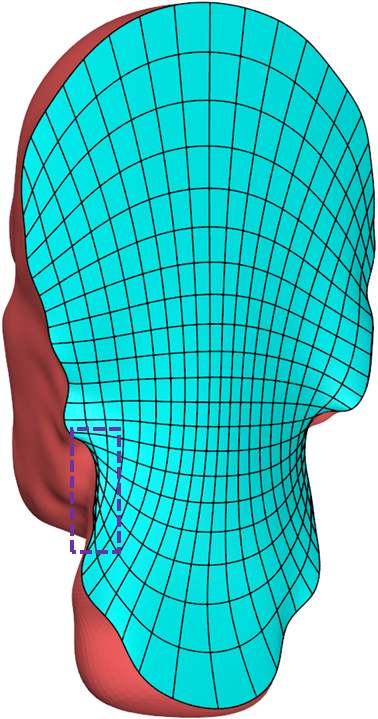}
    }
	\caption{The initial parameterization for the Max Planck model using harmonic mapping.~\subref{HarmonicMapping:BoundaryNurbs} Six boundary NURBS surfaces,~\subref{HarmonicMapping:InteriorOne} to~\subref{HarmonicMapping:InteriorThree} show the interior cross-sections of the parameterization, and the dotted boxes mark the invalid regions.}
    \label{HarmonicMapping}
\end{figure}

\subsection{Construction of a bijective parameterization}
\label{sec:approach:constructbijectiveparameterization}
With the above constructed harmonic map as the input, the next stage of our method is to compute a bijective parameterization of the computational domain $\Omega$. This is the key step of our algorithm.

\subsubsection{Theoretic foundation}
To ensure the bijectivity of the parameterization $\mathbf G$, the sufficient conditions proposed in Section~\ref{sec:BijectiveCondition} were applied by previous literature like~\citep{xu2013analysis} and~\citep{wang2014optimization}.
However, the number of such conditions can be very large, which makes the construction of a bijective parameterization very difficulty.

In this paper, we propose an alternative approach to solve the problem. We maintain a set of collocation points $\mathcal{P}=\{\mathbf{p}_k\}_{k=1}^{N}\subset \hat{\Omega}$, over which we explicitly monitor and control the bijectivity of $\mathbf G$, that is,
\begin{equation}\label{discretebijectivity}
\det(J_{\mathbf G}(\mathbf{p}_k))>0,\quad k=1,\ldots,N.
\end{equation}
Obviously, if the collocation points are dense enough, then the discrete bijectivity constraints~(\ref{discretebijectivity}) can ensure the bijectivity of the map $\mathbf G$ over the whole domain $\hat{\Omega}$. The theoretic foundation of this idea is based on the following lemma and theorem.

\begin{lemma}\label{jacobianlipschitzcontinuous}
The determinant of the Jacobian of a map ${\mathbf G}$ is Lipschitz continuous if
\begin{equation*}
L=\max_{\mathbf{x}\in\hat{\Omega}}\|\nabla \det(J_{\mathbf G}(\mathbf{x}))\|<+\infty.
\end{equation*}
In fact,
\begin{equation}
|\det(J_{\mathbf G}(\mathbf{x}))-\det(J_{\mathbf G}(\mathbf{y}))|\leq L\|\mathbf{x}-\mathbf{y}\|,\quad \forall \mathbf{x},\mathbf{y}\in \hat{\Omega}.
\end{equation}
\end{lemma}

This lemma can be proved using the mean value theorem and it relates the bijectivity of the map $\mathbf{G}$ over the parametric domain $\hat\Omega$ to the bijectivity on the set of collocation points $\mathcal{P}$. Specifically, we have the following theorem.
\begin{theorem}\label{injectivitycontrol}
Given a set of collocation points $\mathcal{P}$ and a positive lower bound $\delta$ of the determinant of the Jacobian of the map $\mathbf{G}$ over $\mathcal{P}$, the bijectivity of $\mathbf{G}$ can be accomplished if the fill distance $d(\mathcal{P},\hat{\Omega})$ of the collocation points in $\hat{\Omega}$ satisfies
\begin{equation*}
d(\mathcal{P},\hat{\Omega})<\frac{\delta}{L},
\end{equation*}
where the fill distance $d(\mathcal{P},\hat{\Omega})$ is the one-sided Hausdorff distance:
\begin{equation*}
d(\mathcal{P},\hat{\Omega})=\max_{\mathbf{x}\in \hat{\Omega}}\min_{\mathbf{p}\in\mathcal{P}}\|\mathbf{x}-\mathbf{p}\|.
\end{equation*}
\end{theorem}
\begin{proof}
For $\forall \mathbf{x}\in\hat{\Omega}$, $\exists \mathbf{p}\in\mathcal{P}$ such that $d(\mathcal{P},\hat{\Omega})<\delta/L$.
According to Lemma~\ref{jacobianlipschitzcontinuous}, we have
\begin{equation*}
\|\det(J_{\mathbf{G}}(\mathbf{x}))-\det(J_{\mathbf{G}}(\mathbf{p}))\|\leq L\|\mathbf{x}-\mathbf{p}\|\leq Ld(\mathcal{P},\hat{\Omega})< \delta,
\end{equation*}
Therefore,
\begin{equation*}
\det(J_{\mathbf{G}}(\mathbf{x}))\geq \det(J_{\mathbf{G}}(\mathbf{p}))-L\|\mathbf{x}-\mathbf{p}\|>\delta-\delta=0.
\end{equation*}
Thus the conclusion follows by Lemma~\ref{injectivecondition}.
\end{proof}

\subsubsection{Max-min optimization model}
Now we propose a mathematical model to compute a bijective map $\mathbf G$ from $\hat\Omega$ to $\Omega$:

\begin{equation}\label{optimization_bijectivity}
\begin{split}
    \max\limits_{\mathbf{G}}\min\limits_{k} \quad &\det(J_{\mathbf{G}}(\mathbf{p}_k))\\
    \mathrm{s.t.}\quad &\det(J_{\mathbf{G}}(\mathbf{p}_k))\geq \delta,\quad \mathbf{p}_k\in\mathcal{P},\quad k=1,\ldots,N,\\
    &\mathcal{E}(\mathbf{G})\leq \epsilon,\\
    &\mathbf{G}\big|_{\partial \hat{\Omega}}\ \text{is given},
\end{split}
\end{equation}
where $\mathcal{E}(\mathbf{G})$ is a thin-plate energy defined as
\begin{equation*}
\begin{split}
    \mathcal{E}(\mathbf{G}) =\int_{\hat{\Omega}}D_{fair}(\mathbf{G};\mathbf{x})\mathrm{d}\xi\mathrm{d}\eta\mathrm{d}\zeta
\end{split}
\end{equation*}
which is used to control the fairness of the map. $\delta$ and $\epsilon$ are two positive thresholds. In our model,
we maximize the minimum of the determinants of the Jacobian $\det(J_{\mathbf{G}}(\mathbf{p}_k))$ at the collocation points,
which is equivalent to minimize the volume distortion in a sense.

By introducing an auxiliary slack variable $t$, the above max-min optimization problem~(\ref{optimization_bijectivity}) can be converted into the following optimization problem:
\begin{equation}\label{optimization_bijectivityequal}
\begin{split}
    \max\limits_{\mathbf{G},t}\quad &t-\lambda \mathcal{E}(\mathbf{G})\\
    \mathrm{s.t.}\quad &\det(J_{\mathbf{G}}(\mathbf{p}_k))\geq t,\quad \mathbf{p}_k\in\mathcal{P},\quad k=1,\ldots,N,\\
    &t\geq \delta,\\
    &\mathbf{G}\big|_{\partial \hat{\Omega}}\ \text{is given}.
\end{split}
\end{equation}
where $\lambda>0$ is a some free parameter.

\begin{remark}
\rm If we apply the sufficient conditions described in Section~\ref{sec:BijectiveCondition} as constraints to ensure the bijectivity of the map, then our mathematical takes the similar form:
\begin{equation}
\label{optimization_remark}
\begin{split}
    \max\limits_{\mathbf{G}}\min\limits_{k} \quad &B_{k}\\
    \mathrm{s.t.}\quad &B_{k}\geq \delta, \quad k=1,\ldots,\widetilde{N},\\
    &\mathcal{E}(\mathbf{G})\leq \epsilon,\\
    &\mathbf{G}\big|_{\partial \hat{\Omega}}\ \text{is given},
\end{split}
\end{equation}
where the number of the inequality constraints $\widetilde{N}$  is generally much larger than $((l-p+1)(3p-1)+1)\times((m-q+1)(3q-1)+1)\times((n-r+1)(3r-1)+1)\approx 27pqrlmn$, and it is much more than that in the model~(\ref{optimization_bijectivity}).
Table~\ref{ConstraintsComparisons} presents a comparison of the number of constraints in the two mathematical models for five examples. The number of constraints in the model (\ref{optimization_bijectivity}) is only about one fourth of that by the model (\ref{optimization_remark}). Thus our optimization  model~(\ref{optimization_bijectivity}) is much less computationally expensive than the model (\ref{optimization_remark}).
\end{remark}

\begin{table*}[!htbp]
	\begin{center}
		\scriptsize
        \renewcommand{\arraystretch}{1.5}
		\begin{tabular}{|c|c|c|}\hline
			\multirow{2}{*}{Model}&\multicolumn{2}{c|}{\#Constraints}\\
            \cline{2-3}
            &(\ref{optimization_bijectivity})&(\ref{optimization_remark})\\
			\hline
           Tooth (Fig.~\ref{ToothParameterization})&99,591&417,115\\
           \hline
           Duck (Fig.~\ref{DuckParameterization})&83,792&303,644\\
           \hline
           Lion (Fig.~\ref{LionParameterization})&133,025&397,781\\
           \hline
           Femur (Fig.~\ref{FemurParameterization})&204,034&834,230\\
           \hline
           Max Planck (Fig.~\ref{MaxPlanckParameterization})&155,939&510,378\\
           \hline
		\end{tabular}
		\caption{Statistics of the number of constraints in the optimization models~(\ref{optimization_bijectivity}) and~(\ref{optimization_remark}). For convenience, we only count the number of constraints for~(\ref{optimization_remark}) at the first level (i.e., the level after the harmonic mapping), that is, the numbers listed here are the least numbers for (\ref{optimization_remark}); whereas the number of constraints for~(\ref{optimization_bijectivity}) is at the last level.}
		\label{ConstraintsComparisons}	
	\end{center}
\end{table*}

\subsubsection{A coarse-to-fine parameterization method}
\label{sec:approach:coarsetofine}

A key issue to solve the optimization problem (\ref{optimization_bijectivityequal}) is to select the collocation points appropriately. Too many collocation points introduce too many nonlinear constraints in the optimization problem, and computational complexity of the problem increases tremendously.  To accelerate the computational process, we propose a coarse-to-fine strategy with a Jacobian gradient guided selection scheme for the collocation points. Specifically, initialize $\mathcal{P}$ as the empty set and ${\mathbf G}$ as the harmonic map computed in Section~\ref{subsection:initialmap}, our algorithm repeats the following two steps:
\begin{enumerate}
\item Convert the B-spline representation of $\det(J_{{\mathbf{G}}})$ into B{\'e}zier forms, and for each B{\'e}zier polynomial $B^\gamma$ check the positivity of $B^\gamma$ by looking at its B{\'e}zier coefficients or those of its subdivided B{\'e}zier polynomials (to certain level) as described in Section~\ref{sec:BijectiveCondition}. If all the B{\'e}zier coefficients are above $\delta (>0)$, then $B^\gamma$ is positive over its defining cell $\hat\Omega_\gamma$ and we remove the collocation points (if any) in $\hat\Omega_\gamma$ from $\mathcal P$; otherwise add new collocation points from the cell $\hat\Omega_\gamma$ to $\mathcal P$. See Figures~\ref{CellSubdivison} and~\ref{drawmesh} for an illustration in 2D case.
\item Compute a parameterization $\mathbf{G}$ by solving the problem~(\ref{optimization_bijectivityequal}) using the algorithm presented in Section~\ref{numericalalgorithm}.
\end{enumerate}

The above procedure runs until the parameterization becomes bijective or the level counter reaches the maximal level $L_{\max}$. The bijectivity of the parameterization is achieved if the bijectivity conditions for all the cells are satisfied.

\begin{figure}[!ht]
\centering
\includegraphics[width=0.9\textwidth]{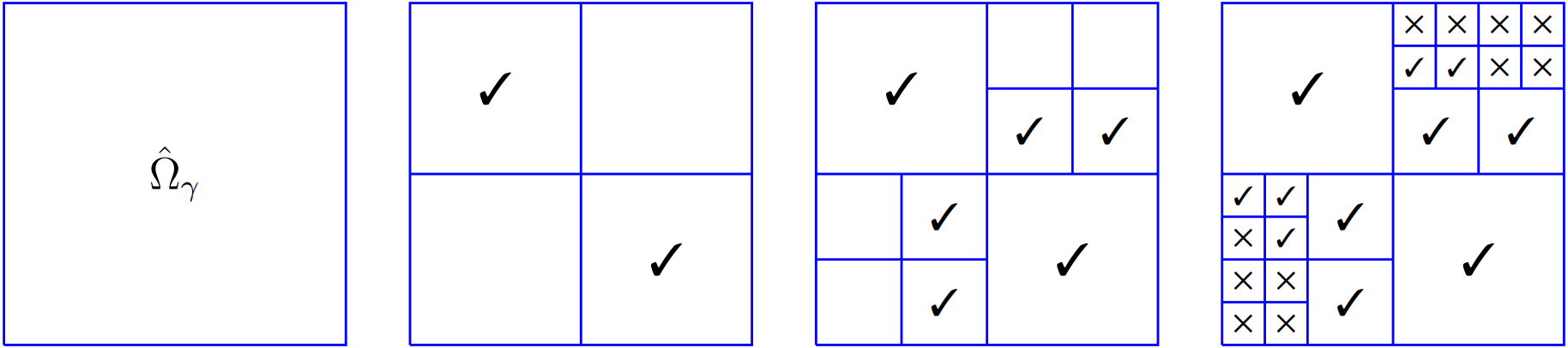}
\caption{Checking the positivity of a B{\'e}zier polynomial by subdivision: the sub-cell marked with `\ding{51}' or `\ding{53}' denotes the positivity of the corresponding B{\'e}zier polynomial is or is not satisfied respectively.}
\label{CellSubdivison}
\end{figure}

\begin{figure}[!htbp]
    \centering
    \begin{tabular}{c@{\hspace{0.05\textwidth}}c@{\hspace{0.05\textwidth}}c}
 	\includegraphics[width=0.22\textwidth]{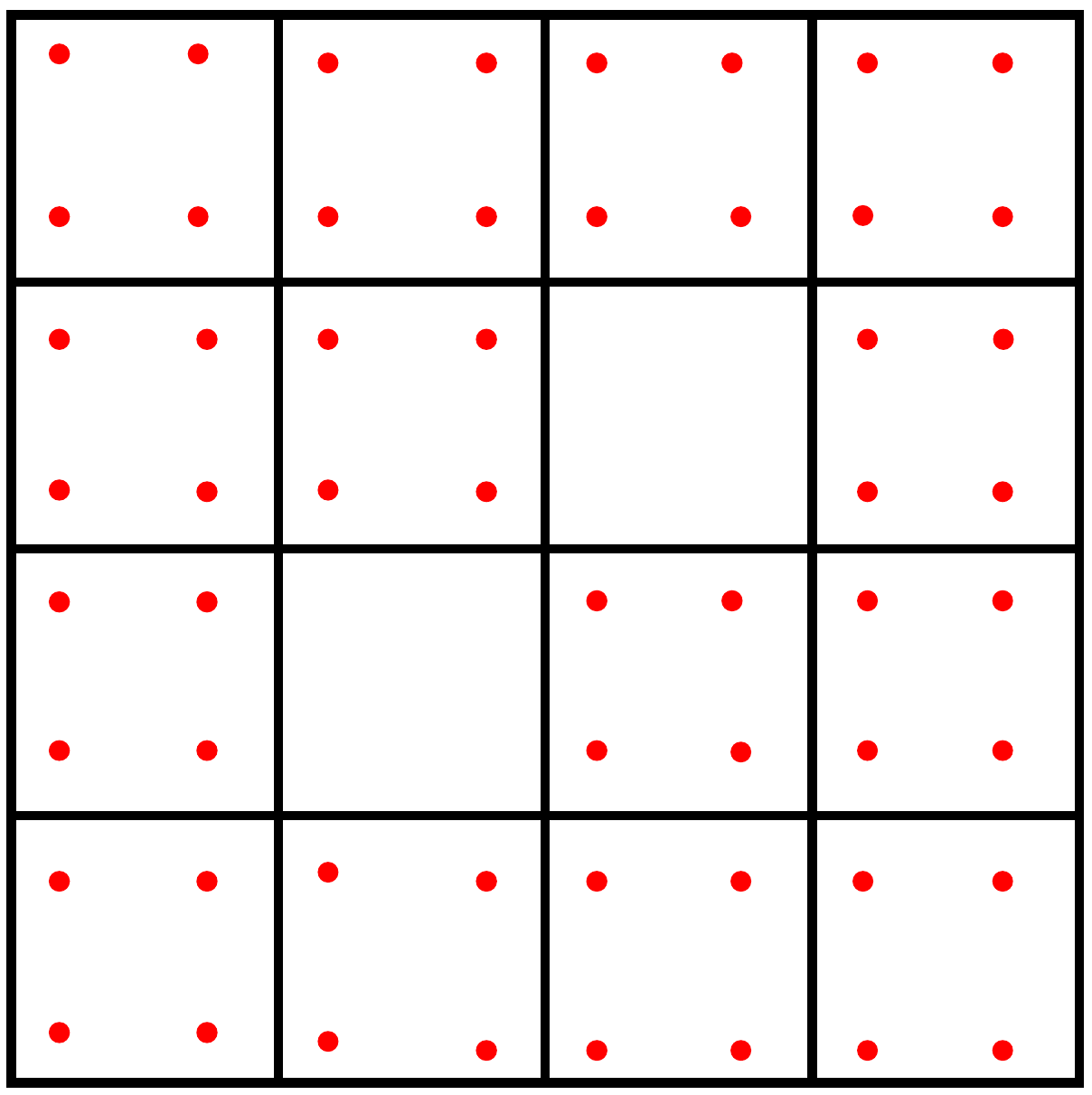}&
    \includegraphics[width=0.22\textwidth]{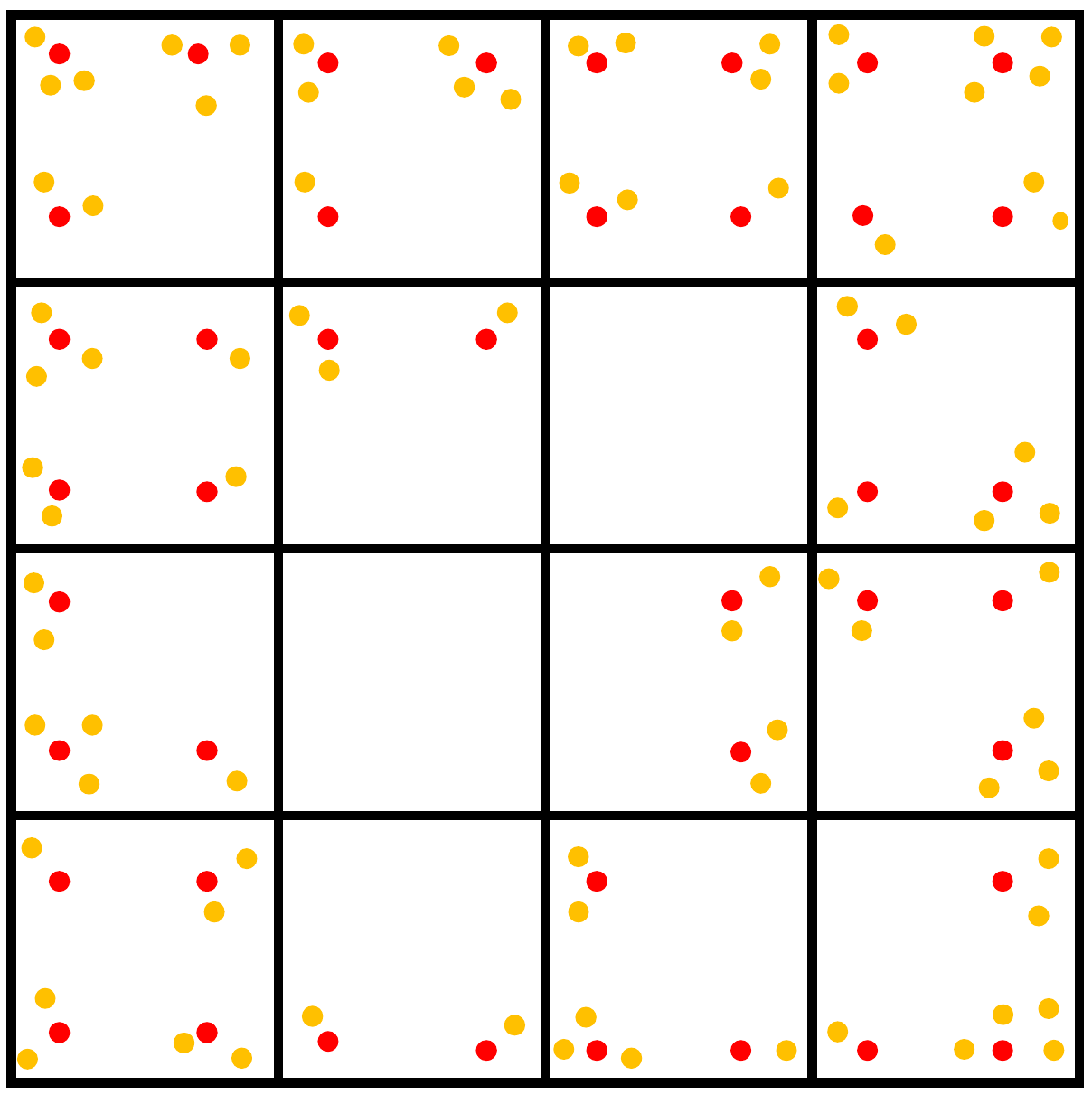}&
    \includegraphics[width=0.22\textwidth]{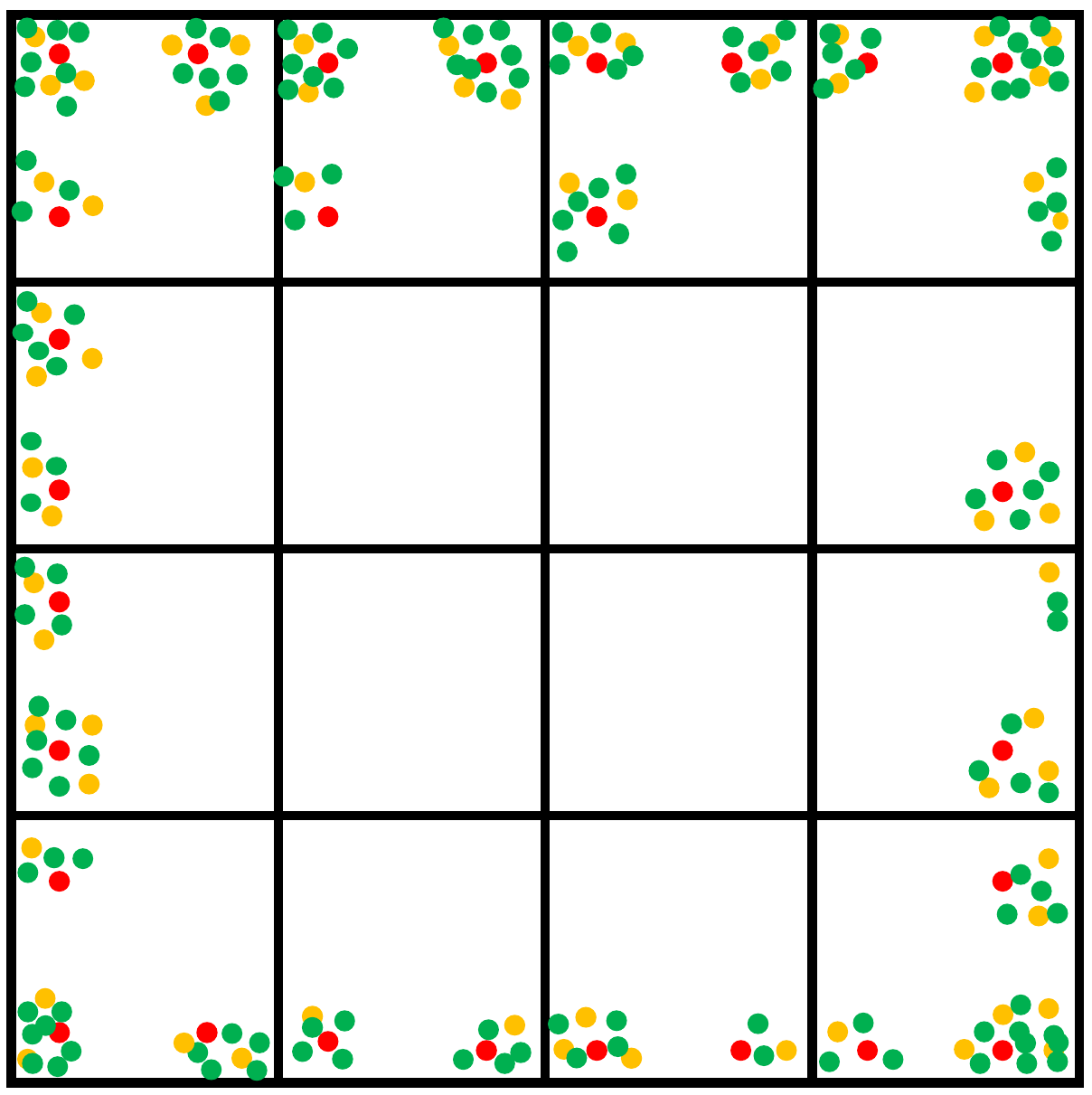}\\
    \end{tabular}
	\caption{An illustration of generating collocation points in 2D: from left to right: collocation points at level $0$, level $1$ and level $2$. The red, orange and green ones denote the new generated points of level $0$, $1$, $2$ respectively.}
    \label{drawmesh}
\end{figure}

Over the cell $\hat{\Omega}_\gamma$ at level $l$, obviously we need only to add the new collocation points in those sub-cells marked with `\ding{53}' (see Fig.~\ref{CellSubdivison}), and these new collocation points are generated based on the gradient of $\det(J_{\mathbf G})$, i.e., more points are needed where the gradient is large while fewer ones where the gradient is small. Specifically, for a sub-cell marked with `\ding{53}', denoted as $\hat{\Omega}_{\gamma}^s$, we firstly uniformly subdivide it into $L\times M\times N$ sub-cuboids $\{\hat{\Omega}_{\gamma_{ijk}}^s\}$, ($i=1,\ldots,L$, $j=1,\ldots,M$, $k=1,\ldots,N$), then the number of new collocation points added to each sub-cuboid $\hat{\Omega}_{\gamma_{ijk}}^s$ can be computed as follows

\begin{equation*}
\hat{N}_{\gamma_{ijk}}^s=\left\lceil2^{3l+9}\cdot\frac{\text{Vol}(\hat{\Omega}_\gamma^s)}{\text{Vol}(\hat{\Omega}_\gamma)}\cdot\frac{\exp{\left(-\frac{\left(g_{\gamma_{ijk}}^s
-g_{\text{max}}\right)^2}{\sigma}\right)}}{\sum\limits_{i,j,k} \exp{\left(-\frac{\left(g_{\gamma_{ijk}}^s-g_{\text{max}}\right)^2}{\sigma}\right)}}\right\rceil
\end{equation*}
where $g_{\gamma_{ijk}}^s$ is the value of $|\nabla \det(J_{\mathbf{G}})|$ at the center of $\hat{\Omega}_{\gamma_{ijk}}^s$, $g_{\max}$ is the maximum value of $\{g_{\gamma_{ijk}}^s\}$, $\sigma$ is an adjusting positive parameter, and $\text{Vol}(\hat{\Omega}_\gamma^s)$, $\text{Vol}(\hat{\Omega}_\gamma)$ denote the volumes of $\hat{\Omega}_\gamma^s$, $\hat{\Omega}_\gamma$ respectively.
Fig.~\ref{BijectiveProcess} illustrates  the whole process of our parameterization method.

\begin{figure}[!htbp]
    \centering
    \subfigure[Boundary NURBS]{
        \label{BijectiveProcess:BoundaryNurbs}
        \includegraphics[width=0.19\textwidth]{boundarynurbsmaxplanck15}
    }
    \hspace{0.07\textwidth}
    \subfigure[Harmonic mapping]{
        \label{BijectiveProcess:HarmonicMapping}
        \includegraphics[width=0.18\textwidth]{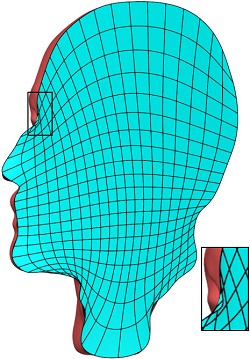}
    }
    \hspace{0.07\textwidth}
    \subfigure[Level 0]{
        \label{BijectiveProcess:LevelOne}
        \includegraphics[width=0.18\textwidth]{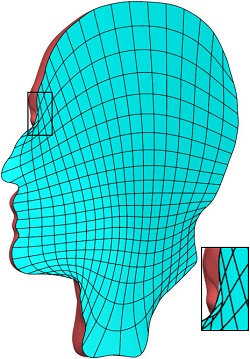}
    }
    \hspace{0.07\textwidth}
    \subfigure[Level 1]{
        \label{BijectiveProcess:LevelTwo}
        \includegraphics[width=0.175\textwidth]{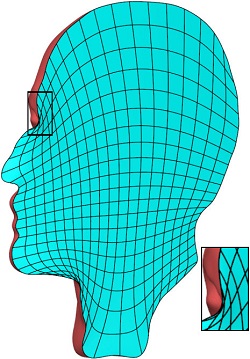}
    }
    \hspace{0.07\textwidth}
    \subfigure[Level 2]{
        \label{BijectiveProcess:LevelThree}
        \includegraphics[width=0.175\textwidth]{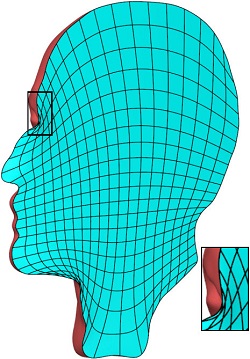}
    }
    \hspace{0.07\textwidth}
    \subfigure[Level 3]{
        \label{BijectiveProcess:LevelFour}
        \includegraphics[width=0.175\textwidth]{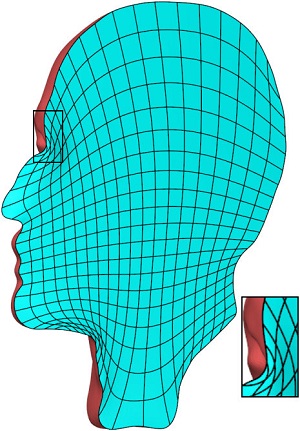}
    }
	\caption{The process of constructing a bijective parameterization for the Max Planck model.~\subref{BijectiveProcess:BoundaryNurbs} Six boundary spline surfaces,~\subref{BijectiveProcess:HarmonicMapping} the initial harmonic mapping,~\subref{BijectiveProcess:LevelOne} to~\subref{BijectiveProcess:LevelFour} show the intermediate parameterizations from level 0 to level 3, and~\subref{BijectiveProcess:LevelFour} is a bijective parameterization.}
    \label{BijectiveProcess}
\end{figure}

\subsubsection{Numerical algorithm}\label{numericalalgorithm}
 The optimization problem~(\ref{optimization_bijectivityequal}) can be solved by the interior-point algorithm~\citep{nocedal2006numerical} which is effective in solving large-scale nonlinear optimization problems.
However, when parameterizing complex domains, the number of variables and constraints in~(\ref{optimization_bijectivityequal}) can be very large, ranging from tens to hundreds of thousands. So many variables and constraints are a big burden for solving nonlinear constrained optimization problems. In order to improve the computational efficiency, we apply a divide and conquer strategy to solve the problem~(\ref{optimization_bijectivityequal}). Generally speaking, the main idea is to split
the parametric domain into several non-overlapping regions and then solve the problem~(\ref{optimization_bijectivityequal})
over each subregion separately, and finally combine the solutions over the subregions to constitute the solution of the original problem. Fig.~\ref{DivideAndConquer} illustrates this approach in 2D case, where the parametric domain $\hat{\Omega}$ is subdivided into eight sub-domains (denoted as $\widetilde{\Omega}_i$, $i=1,\ldots,8$), and the problem~(\ref{optimization_bijectivityequal}) is solved over each sub-domain $\widetilde{\Omega}_i$ separately.

\begin{figure}[!ht]
\begin{center}
\includegraphics[width=0.24\textwidth]{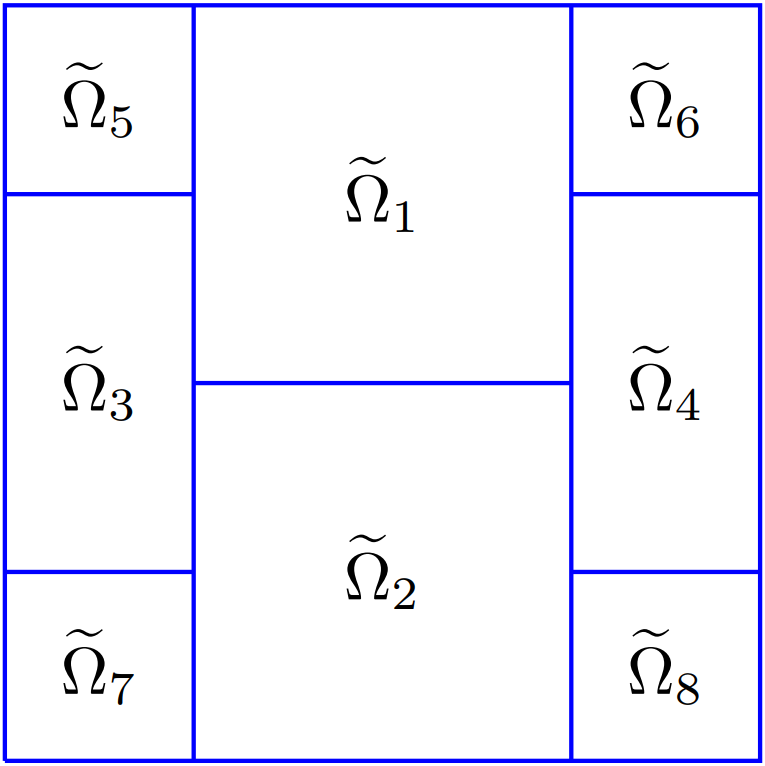}
\caption{\label{DivideAndConquer} An illustration of the divide and conquer approach in 2D.}
\end{center}
\end{figure}

\begin{remark}
In solving the problem~(\ref{optimization_bijectivityequal}) at certain level,  there is no need to compute those control points of the B-spline parameterization (\ref{B-spline rep}) whose influence region is contained in some region marked with `\ding{51}', and they are kept unchanged from previous level. The free variables are the control points whose influence region contains the region marked with `\ding{53}'.
\end{remark}

\begin{remark}
Since B-spline basis functions have local compact supports, only a few variables may contribute to each constraint in~(\ref{optimization_bijectivityequal}), i.e., the constraint conditions are sparse in variables.
This can greatly reduce the computational complexity in our algorithm.
\end{remark}

\begin{remark}
For two adjacent subregions $\widetilde\Omega_i$ and $\widetilde\Omega_j$, the optimization problem (\ref{optimization_bijectivityequal})
may have common variables. In order to keep these sub-problems separable, we solve them in a certain order, such as from the subregion $\widetilde\Omega_1$ to $\widetilde\Omega_8$, and when a sub-problem is solved, its optimization variables are fixed. However, in some situations, there may be no feasible solutions to the subsequent sub-problems due to the lack of degree of freedom. In this case, a local offset function represented by a linear combination of some B-spline basis functions with more compact supports can be added over the corresponding subregion to achieve a feasible solution.
\end{remark}

\subsection{Improving parameterization quality using MIPS}

In the previous subsection, we designed a mathematical model to construct a bijective parameterization of a three dimensional domain. In that model, only volumetric distortion is considered, and consequently the distortion of the parameterization is a bit large in some concave regions, see for example, Fig.~\ref{BijectiveProcess:LevelFour} and Fig.~\ref{QualityImprovement}.
In this subsection, we further improve the quality of the parameterization map constructed in the previous subsection using MIPS--an efficient global parameterization method which minimizes the conformal distortion of the map~\citep{hormann1999mips}.
The problem is formulated as
\begin{equation}\label{qualityimprovementusingmips}
\begin{split}
    \min\limits_{\mathbf{G}} \quad &\int_{\hat{\Omega}}D^2_{con}(\mathbf{G};\mathbf{x})\mathrm{d}\mathbf{x}\\
    \mathrm{s.t.}\quad &\mathbf{G}\big|_{\partial \hat{\Omega}}\ \text{is given},
\end{split}
\end{equation}
where $D_{con}(\mathbf{G};\mathbf{x})$ is the conformal distortion of $\mathbf{G}$ introduced in Section~\ref{sec:distortiondefinition}.

To solve the above optimization problem, one first has to discretize the integral using Gaussian quadrature rules.
Then the problem can be solved efficiently using the L-BFGS which is a quasi-Newton method for solving unconstrained nonlinear minimization problems~\citep{nocedal2006numerical}. The initial solution is provided by the bijective parameterization computed in the last subsection. The bijectivity of the parameterization can be preserved during the optimization process since MIPS can penalize invalid parameterization. In fact the objective function in~(\ref{qualityimprovementusingmips}) goes to infinity when $\det(J_{\mathbf{G}}(\mathbf{x}))$ approaches zero. Numerical examples demonstrate that this step of optimization is essential and effective, see Fig.~\ref{QualityImprovement} for a comparison.
\begin{figure}[!htbp]
    \centering
    \begin{tabular}{c@{\hspace{0.1\textwidth}}c}
    \includegraphics[width=0.25\textwidth]{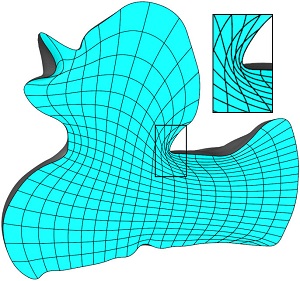}&
    \includegraphics[width=0.25\textwidth]{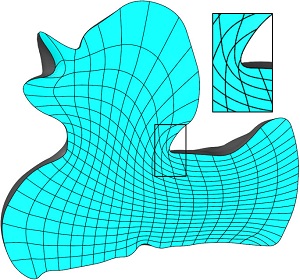}\\
    \includegraphics[width=0.25\textwidth]{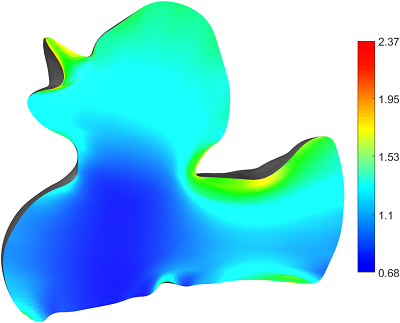}&
    \includegraphics[width=0.25\textwidth]{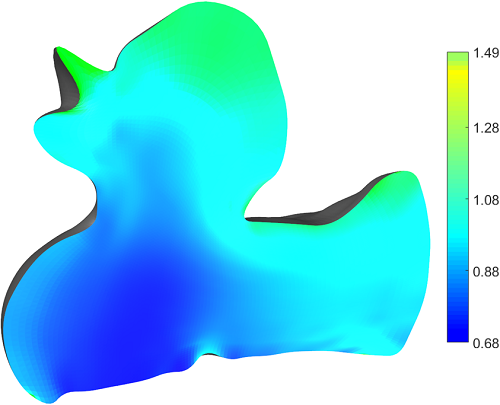}\\
    \end{tabular}
	\caption{Comparison of parameterization results with (right column) and without (left column) MIPS. The top row shows the iso-parametric surfaces and the bottom row shows the colormaps of $\log_5\kappa(J_{\mathbf{G}})$, where $\kappa(J_{\mathbf{G}})$ characterizes the conformal distortion of the parameterization which will be introduced in Section~\ref{sec:result:qualitymetric}.}
    \label{QualityImprovement}
\end{figure}

\smallskip

Now the overall algorithm of our volumetric parameterization is summarized in Algorithm~\ref{alg:PGVPA}.

\begin{algorithm}[H]
\begin{algorithmic}[1]
	\caption{Volumetric parameterization algorithm} \label{alg:PGVPA}
    \REQUIRE ~~\\ 
    The B-spline representations of the six boundary surfaces of the domain $\Omega$, and the parameters $\lambda$ and $\delta$.
    \ENSURE ~~\\ 
   A bijective map $\mathbf{G}$ with low distortion from the unit cube $\hat{\Omega}$ to $\Omega$.\\
    \STATE {Compute a harmonic map by solving the optimization problem~(\ref{harmonicmap}).}\\
	\STATE {With the harmonic map as the initialization, construct a bijective parameterization by solving
            the max-min optimization problem (\ref{optimization_bijectivityequal}) in a coarse-to-fine way
            together with a divide and conquer strategy.}\\
	\STATE {Based on the obtained bijective parameterization, solve the problem~(\ref{qualityimprovementusingmips}) to further improve the quality of the parameterization.}
\end{algorithmic}
\end{algorithm}

\section{Experimental results}
\label{sec:result}
In this section, we apply our volumetric parameterization method to several computational domains, and for each example, the bijectivity and distortion are demonstrated. Furthermore, comparisons with two nonlinear optimization methods~\citep{xu2013analysis,wang2014optimization} are performed. We didn't compare our method with some other competitive methods~\citep{escobar2011new,zhang2012solid,wang2013trivariate,lopez2017spline} since they construct parameterizations from triangle meshes rather than spline boundary surfaces and they utilized T-splines instead of NURBS as the representation of the parameterization. Furthermore, these methods do not guarantee the bijectivity of the parameterization.

\subsection{Implementation details}
Our method is implemented using C++ with visual studio 2017 on windows 10 platform. The optimization problems~(\ref{optimization_bijectivityequal}) and~(\ref{qualityimprovementusingmips}) are solved using the Artelys Knitro's software~\citep{byrd2006knitro}. All experiments are conducted on a laptop computer of an Intel Core i5-7300HQ CPU, 2.5GHz with 8GB memory. There are several parameters which need to be specified. The knot vectors and the degree of the B-spline representation in~(\ref{B-spline rep}) are provided in Table~\ref{PerformanceStatistics}. In the bijective parameterization model~(\ref{optimization_bijectivityequal}), there are two parameters, i.e., $\lambda$ and $\delta$. We typically set $\delta\in [\text{1E-3}, \text{3E-2}]$. The weight $\lambda$ can be used to balance the bijectivity and distortion of the parameterization. Clearly, larger $\lambda$ leads to parameterization results of lower distortion while smaller $\lambda$ can increase the bijectivity of the parameterization. We observe that $\lambda\in [0.5,3]$ provides a good compromise between them.
\subsection{Parameterization results}
\subsubsection{Quality metric for parameterization}
\label{sec:result:qualitymetric}
 We use three metrics---the condition number of the Jacobian matrices, the orthogonality of the iso-parametric elements and the volume distortion of the parameterizaton  to evaluate the quality of a parameterization.

 The condition number $\kappa(J_{\mathbf{G}})$ of the Jacobian matrix $J_{\mathbf{G}}$ defined as
\begin{equation*}
\kappa(J_{\mathbf{G}})=\|J_{\mathbf{G}}\|_F\|J^{-1}_{\mathbf{G}}\|_F
\end{equation*}
indicates whether the Jacobian matrix is ill-conditioned or not. It's easy to see that this metric can characterize the conformal distortion of the parameterization by (\ref{disconformal}).

The orthogonality of iso-parametric structure is also an important quality measure of analysis suitable volumetric parameterization in IGA applications~\citep{xu2014high}, which can be defined according to the differential geometry property of parametric volumes as follows
\begin{equation*}
\mathbf{G}_{\text{orth}}=\left(1-\left|\frac{\mathbf{G}_{\xi}}{\|\mathbf{G}_{\xi}\|_2}\cdot\frac{\mathbf{G}_{\eta}}{\|\mathbf{G}_{\eta}\|_2}\right|\right)\times\left(1-
\left|\frac{\mathbf{G}_{\eta}}{\|\mathbf{G}_{\eta}\|_2}\cdot\frac{\mathbf{G}_{\zeta}}{\|\mathbf{G}_{\zeta}\|_2}\right|\right)\times
\left(1-\left|\frac{\mathbf{G}_{\zeta}}{\|\mathbf{G}_{\zeta}\|_2}\cdot\frac{\mathbf{G}_{\xi}}{\|\mathbf{G}_{\xi}\|_2}\right|\right).
\end{equation*}
Obviously $\mathbf{G}_{\text{orth}}$ achieves its maximal value $1$ if and only if $\mathbf{G}_{\xi}\cdot\mathbf{G}_{\eta}=\mathbf{G}_{\eta}\cdot\mathbf{G}_{\zeta}=\mathbf{G}_{\zeta}\cdot\mathbf{G}_{\xi}=0$. And the larger $\mathbf{G}_{\text{orth}}$ is, the more orthogonal the iso-parametric structure of the parameterization is.
Again this metric is related with conformal distortion.

Besides the above two metrics, another important criteria to evaluate the quality of the parameterization is the volume distortion $D_{vol}({\mathbf G})$ (also called as the scaled Jacobian) as described in Section~\ref{sec:distortiondefinition}.
%
In our experiments, to measure the volume distortion of the parameterization, we uniformly partition the parametric domain $\hat{\Omega}$ into $\hat{L}\times \hat{M}\times \hat{N}$ sub-cuboids $\{\hat{\Omega}_{ijk}\}$, ($i=1,\ldots,\hat{L}$, $j=1,\ldots,\hat{M}$, $k=1,\ldots,\hat{N}$), then the volume distortion over the sub-cuboid $\{\hat{\Omega}_{ijk}\}$, denoted as $D_{vol}({\mathbf{G}})|_{\hat{\Omega}_{ijk}}$, can be computed as
\begin{equation*}
D_{vol}({\mathbf{G}})|_{\hat{\Omega}_{ijk}}=\frac{\int_{\hat{\Omega}_{ijk}}D_{vol}({\mathbf{G}})\mathrm{d}\xi\mathrm{d}\eta\mathrm{d}\zeta}{\text{Vol}({\hat{\Omega}_{ijk}})}
\end{equation*}
where $\text{Vol}({\hat{\Omega}_{ijk}})$ is the volume of $\hat{\Omega}_{ijk}$.

\begin{figure}[!htbp]
    \centering
    \subfigure[Tooth]{
        \label{ParameterizationInput:Tooth}
        \includegraphics[width=0.20\textwidth]{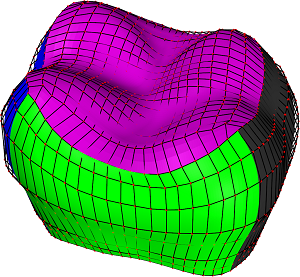}
    }
    \hspace{0.05\textwidth}
    \subfigure[Duck]{
        \label{ParameterizationInput:Duck}
        \includegraphics[width=0.20\textwidth]{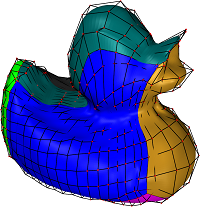}
    }
    \hspace{0.05\textwidth}
    \subfigure[Lion]{
        \label{ParameterizationInput:Lion}
        \includegraphics[width=0.17\textwidth]{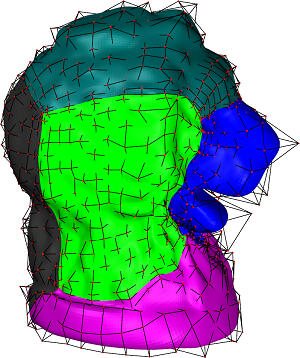}
    }
        \hspace{0.05\textwidth}
    \subfigure[Femur]{
        \label{ParameterizationInput:Femur}
        \includegraphics[width=0.30\textwidth]{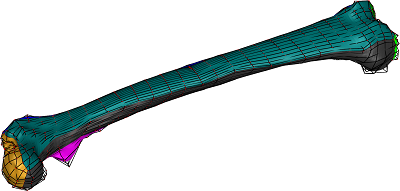}
    }
        \hspace{0.05\textwidth}
    \subfigure[Max Planck]{
        \label{ParameterizationInput:MaxPlanck}
        \includegraphics[width=0.17\textwidth]{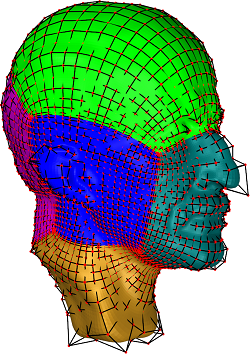}
    }
	\caption{The given boundary B-spline surfaces for five models. ~\subref{ParameterizationInput:Tooth} Tooth,~\subref{ParameterizationInput:Duck} Duck,~\subref{ParameterizationInput:Lion} Lion,~\subref{ParameterizationInput:Femur} Femur and~\subref{ParameterizationInput:MaxPlanck} Max Planck.}
    \label{ParameterizationInput}
\end{figure}

\subsubsection{Parameterization of different computational domains}
We demonstrate five examples (as illustrated in  Fig.~\ref{ParameterizationInput}) to show the effectiveness of the proposed approach, and some comparisons with Xu's method~\citep{xu2013analysis} and Wang's method~\citep{wang2014optimization} are provided.

Fig.~\ref{ToothParameterization} depicts the results of volume parameterization of the tooth model by the three methods. The six given boundary B-spline surfaces are shown in Fig.~\ref{ParameterizationInput:Tooth}. In order to show the parameterization quality, three interior iso-parametric surfaces for Xu's method, Wang's method and our method are provided in the first three rows of Fig.~\ref{ToothParameterization}. The condition number and orthogonality colormaps of these iso-parametric surfaces are shown in the next two rows respectively. In order to verify the volume preserving property of these methods, we calculate the volume distortion distributions, which are illustrated in the last row of Fig.~\ref{ToothParameterization}. The bijectivity of the parameterization for this model is verified by the criterion described in Section~\ref{sec:BijectiveCondition}. From this example it can be seen that our method produces smaller distortion and better orthogonality than the other two methods. Fig.~\ref{DuckParameterization} shows the comparison results of the duck model. In this example, we observe that Xu's method has many self-intersections and the distortion of Wang's method is a bit large in some concave regions, while our method is always bijective and has lower distortion. Fig.~\ref{LionParameterization} shows the parameterization results of the lion model by Wang's method and our method. It can bee seen that Wang's method is not bijective in the extremely concave regions, while our method can still produce bijective parameterization. More results for complex models by our approach are also demonstrated, including the femur model in Fig.~\ref{FemurParameterization} and Max Planck model in Fig.~\ref{MaxPlanckParameterization}. In the last three examples (Fig.~\ref{LionParameterization}-\ref{MaxPlanckParameterization}), Xu's method is unable to achieve a valid parameterization in a limited amount of time. From all these examples, we can conclude that the proposed method can achieve high-quality parameterizations and significantly outperforms the other two methods in terms of bijectivity, distortion and orthogonality.

\begin{figure}[!htbp]
	\centering
	\begin{tabular}{c@{\hspace{0.05\textwidth}}c@{\hspace{0.05\textwidth}}c}
        \includegraphics[width=0.23\textwidth]{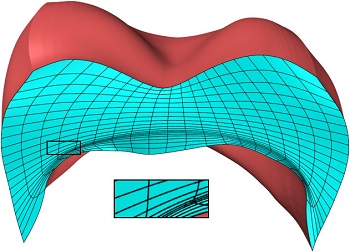}&
        \includegraphics[width=0.23\textwidth]{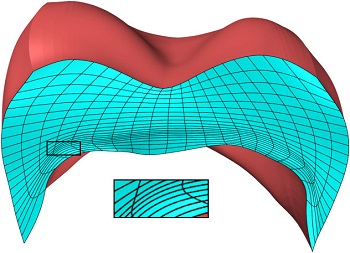}&
         \includegraphics[width=0.23\textwidth]{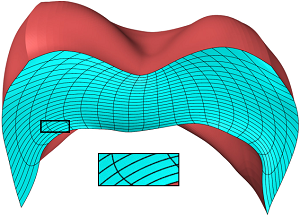}\\
         \includegraphics[width=0.23\textwidth]{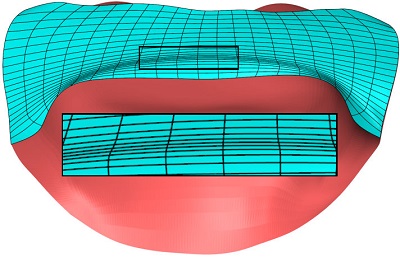}&
        \includegraphics[width=0.23\textwidth]{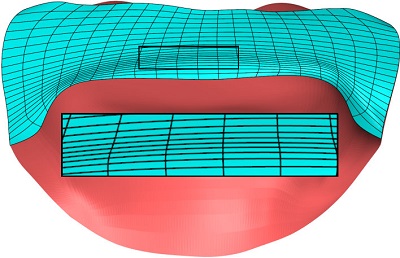}&
         \includegraphics[width=0.23\textwidth]{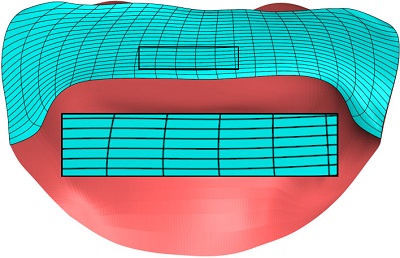}\\
        \includegraphics[width=0.20\textwidth]{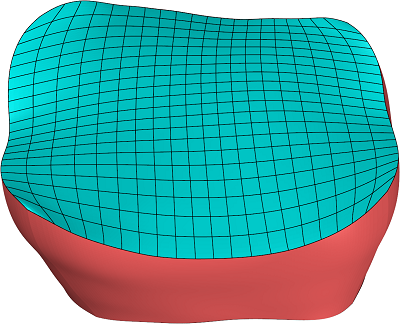}&
        \includegraphics[width=0.20\textwidth]{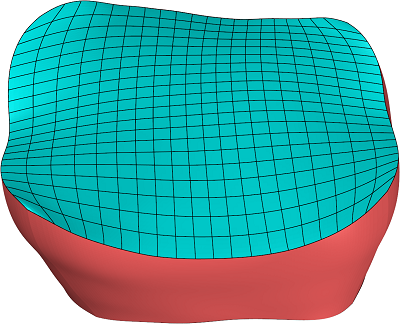}&
         \includegraphics[width=0.20\textwidth]{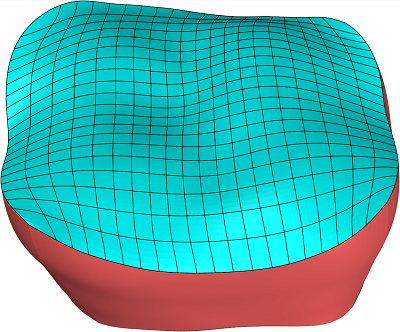}\\
          \includegraphics[width=0.26\textwidth]{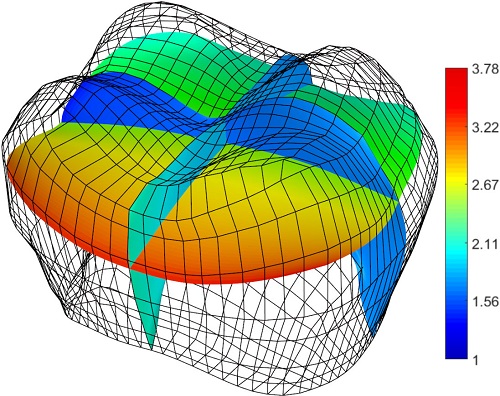}&
        \includegraphics[width=0.26\textwidth]{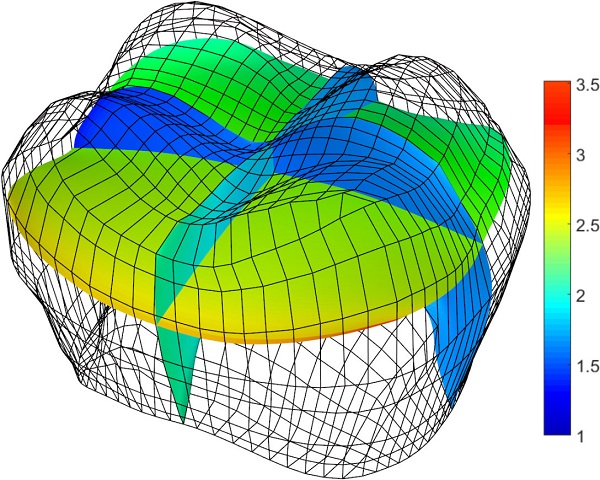}&
         \includegraphics[width=0.26\textwidth]{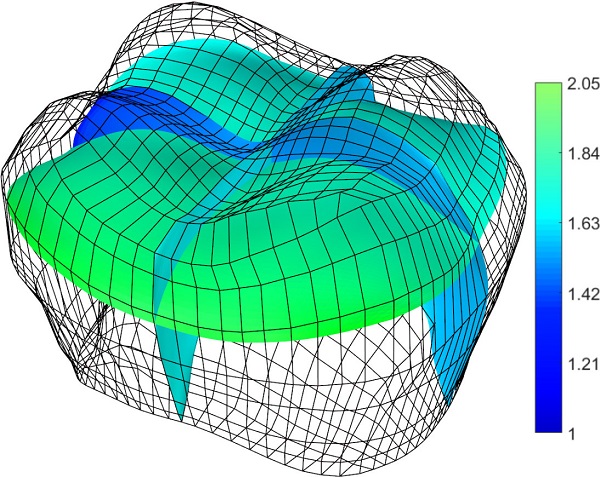}\\
          \includegraphics[width=0.26\textwidth]{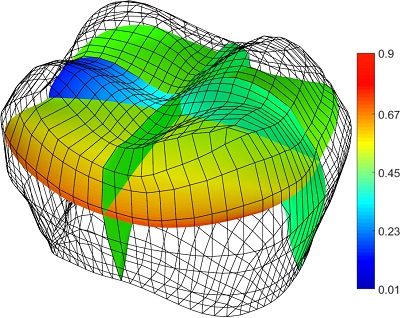}&
        \includegraphics[width=0.26\textwidth]{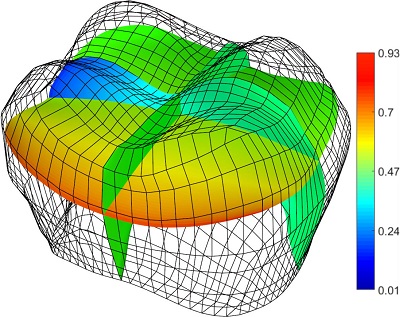}&
         \includegraphics[width=0.26\textwidth]{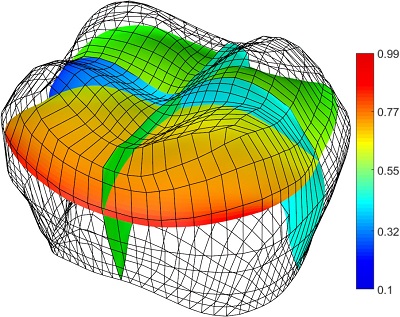}\\
         \includegraphics[width=0.22\textwidth]{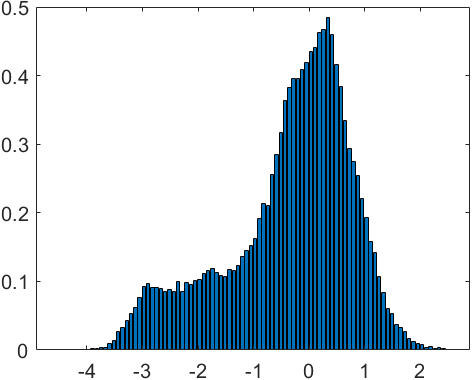}&
        \includegraphics[width=0.22\textwidth]{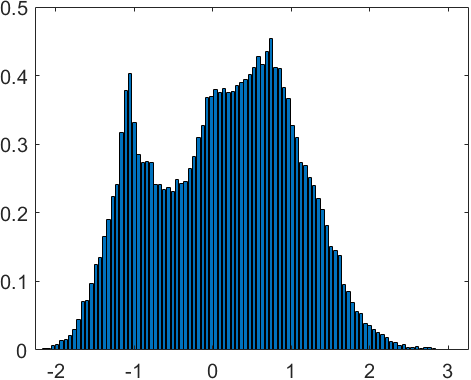}&
         \includegraphics[width=0.22\textwidth]{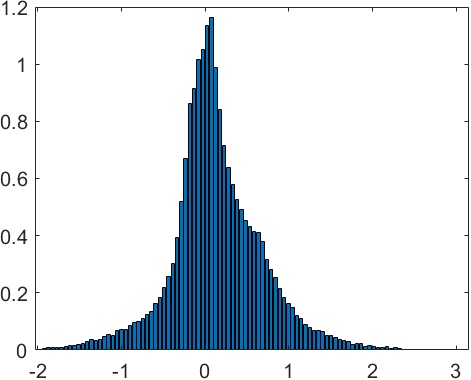}\\
       (a) Xu's & (b) Wang's & (c) Ours
	\end{tabular}
    \caption{Volumetric parameterization of the tooth model by the three methods: (a) Xu's method, (b) Wang's method and (c) our method. The first three rows show the interior iso-parametric surfaces of the parameterization, the next two rows show the colormaps of $\log_3\kappa(J_{\mathbf{G}})$ and $\mathbf{G}_{\text{orth}}$ respectively, and the last row depicts the distributions of volume distortion $\log_2 (D_{vol}(\mathbf{G}))$. Note that the optimal value of $\log_2 (D_{vol}(\mathbf{G}))$ is $0$.}
    \label{ToothParameterization}
\end{figure}

\begin{figure}[!htbp]
	\centering
	\begin{tabular}{c@{\hspace{0.07\textwidth}}c@{\hspace{0.07\textwidth}}c}
        \includegraphics[width=0.21\textwidth]{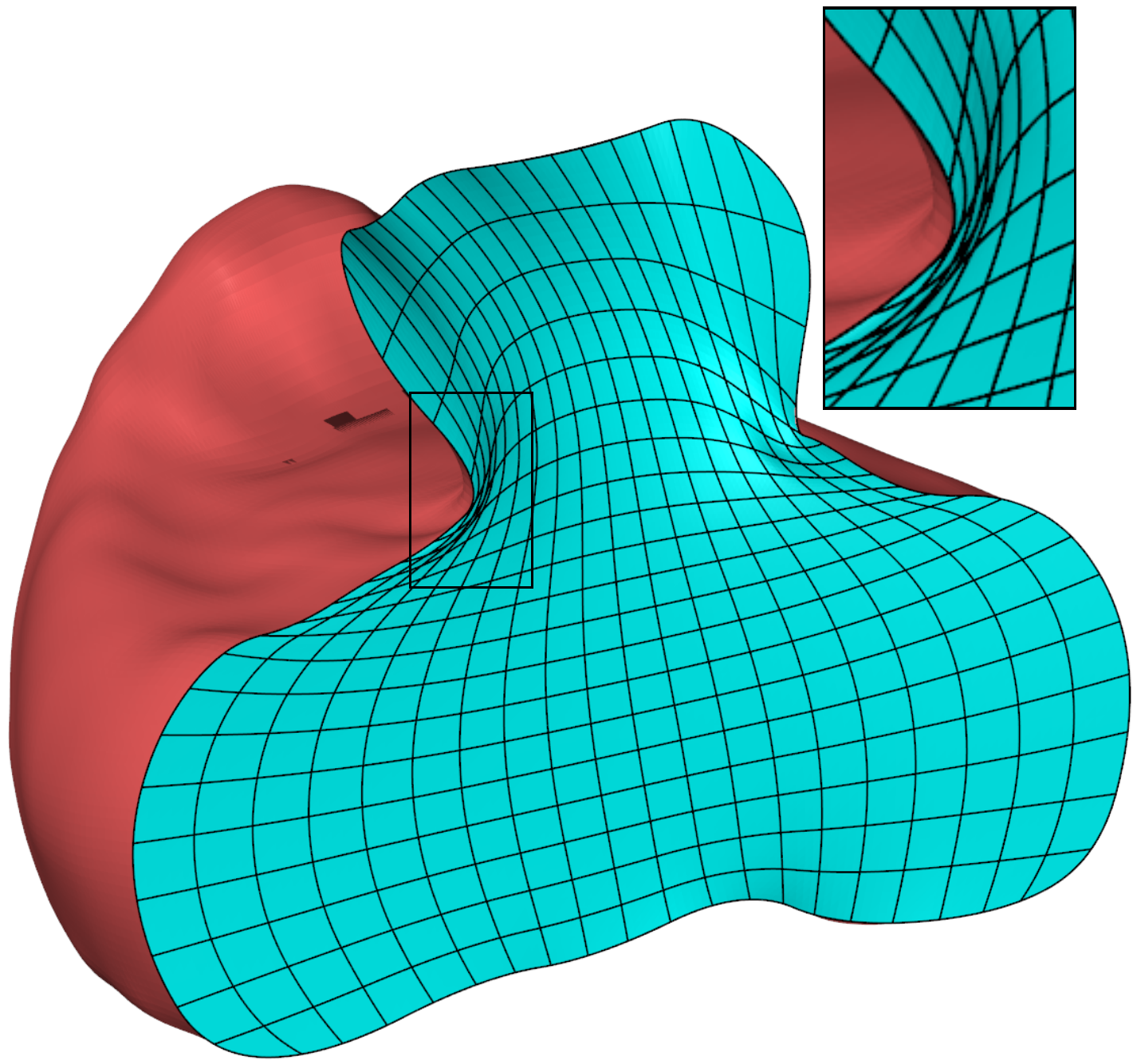}&
        \includegraphics[width=0.21\textwidth]{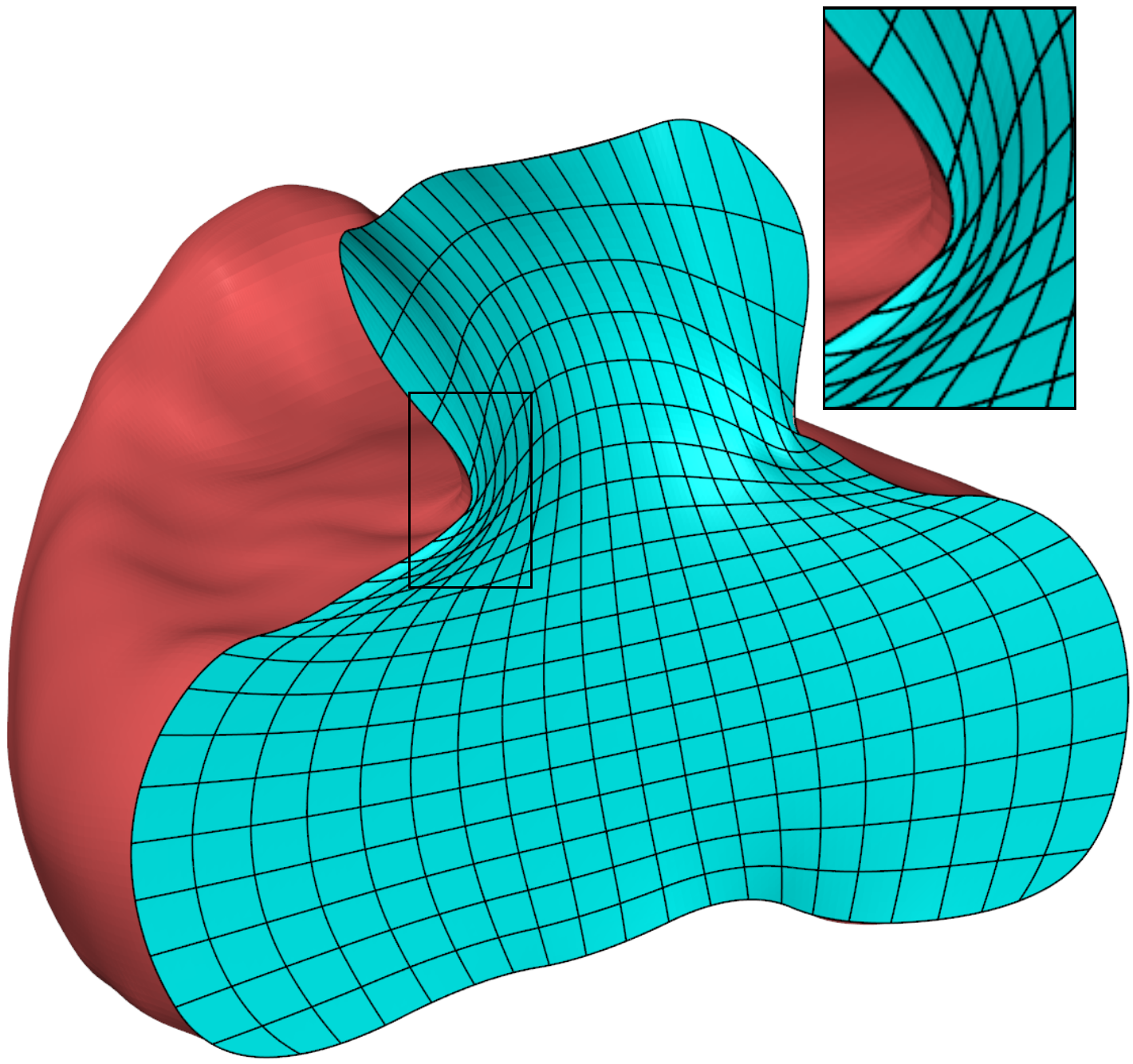}&
         \includegraphics[width=0.21\textwidth]{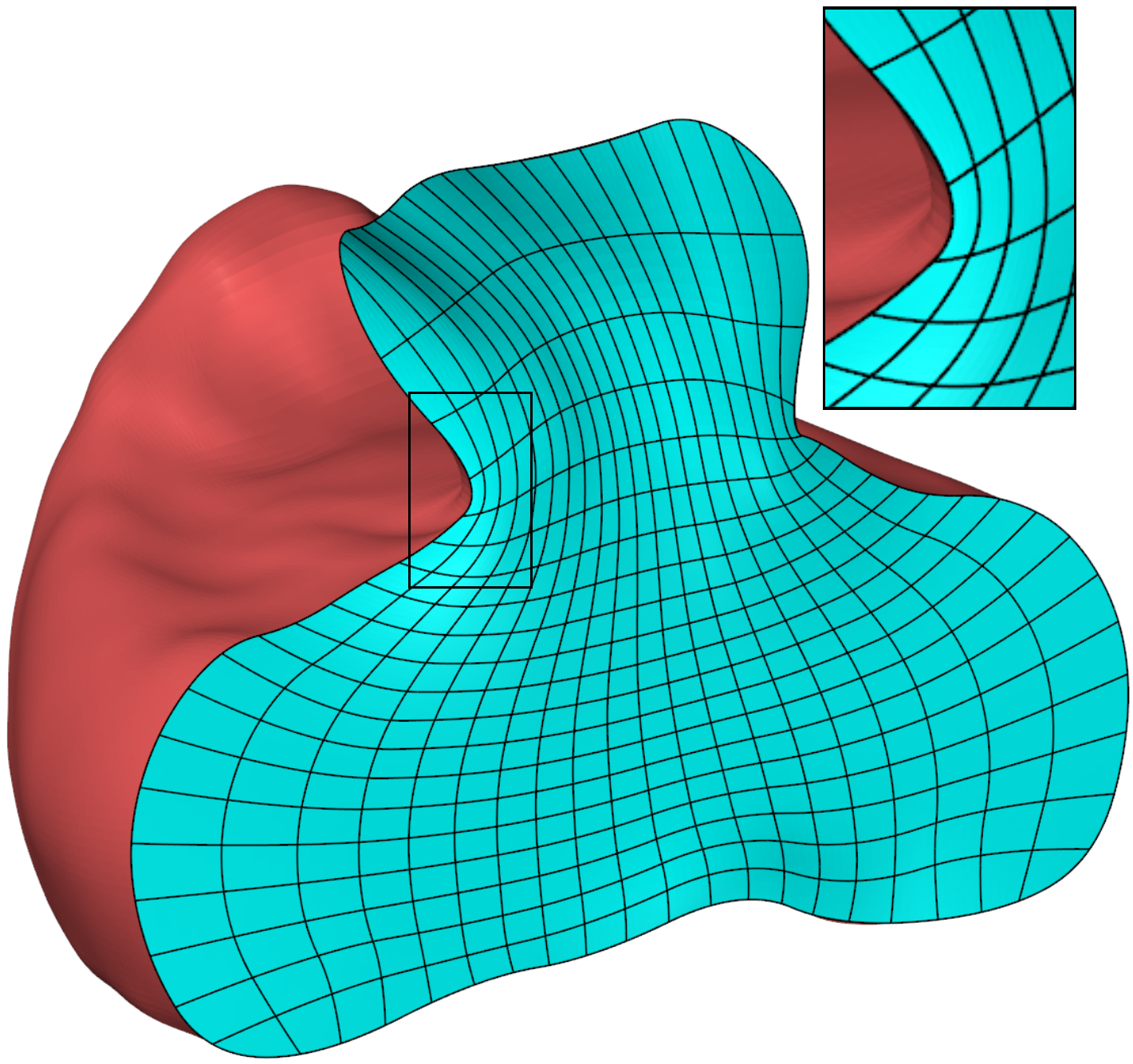}\\
         \includegraphics[width=0.22\textwidth]{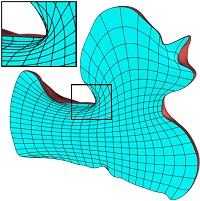}&
        \includegraphics[width=0.22\textwidth]{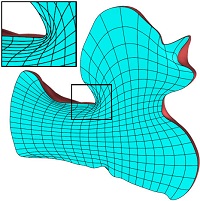}&
         \includegraphics[width=0.22\textwidth]{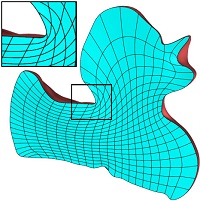}\\
        \includegraphics[width=0.20\textwidth]{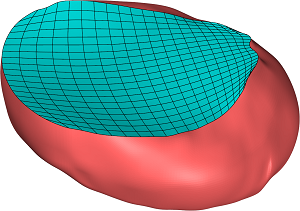}&
        \includegraphics[width=0.20\textwidth]{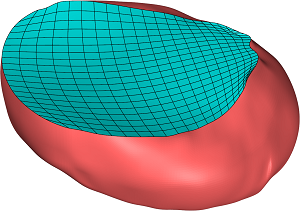}&
         \includegraphics[width=0.20\textwidth]{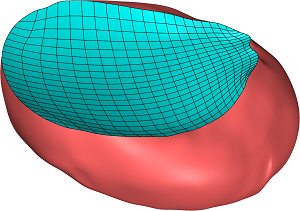}\\
          \includegraphics[width=0.22\textwidth]{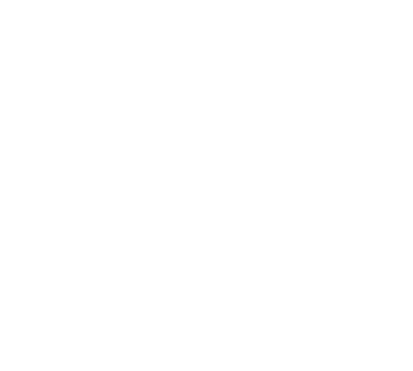}&
        \includegraphics[width=0.22\textwidth]{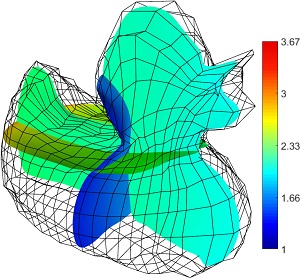}&
         \includegraphics[width=0.22\textwidth]{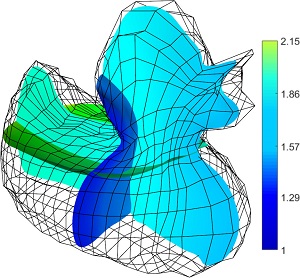}\\
          \includegraphics[width=0.22\textwidth]{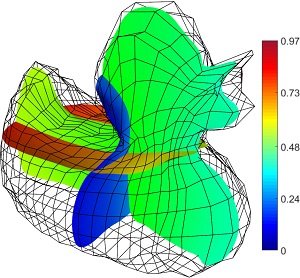}&
        \includegraphics[width=0.22\textwidth]{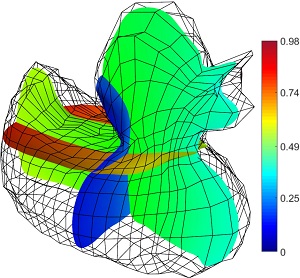}&
         \includegraphics[width=0.22\textwidth]{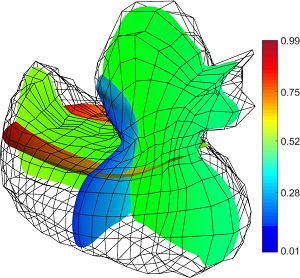}\\
         \includegraphics[width=0.22\textwidth]{blankfig}&
        \includegraphics[width=0.22\textwidth]{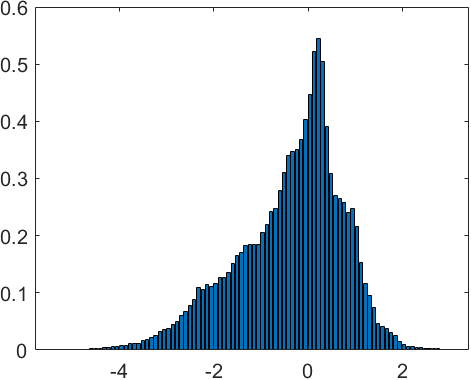}&
         \includegraphics[width=0.22\textwidth]{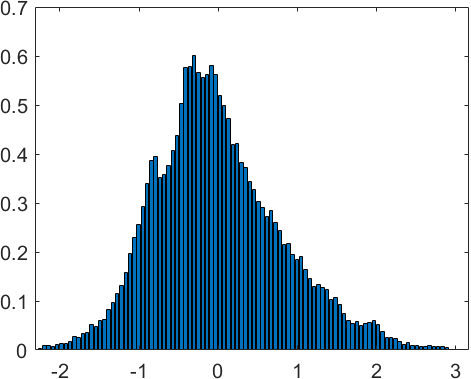}\\
       (a) Xu's & (b) Wang's & (c) Ours
	\end{tabular}
    \caption{Volumetric parameterization of the duck model by the three methods: (a) Xu's method, (b) Wang's method and (c) our method. The first three rows show the interior iso-parametric surfaces of the parameterization, the next two rows show the colormaps of $\log_3\kappa(J_{\mathbf{G}})$ and $\mathbf{G}_{\text{orth}}$ respectively, and the last row depicts the distributions of volume distortion $\log_2 (D_{vol}(\mathbf{G}))$. Here we omit the colormaps of $\log_3\kappa(J_{\mathbf{G}})$ and $\log_2 (D_{vol}(\mathbf{G}))$ for Xu's method due to its invalid parameterization.}
    \label{DuckParameterization}
\end{figure}

\begin{figure}[!htbp]
	\centering
\subfigure[Wang's]{
\label{LionParameterization:Wangs}
\begin{minipage}{0.18\linewidth}
 \includegraphics[width=3.0cm]{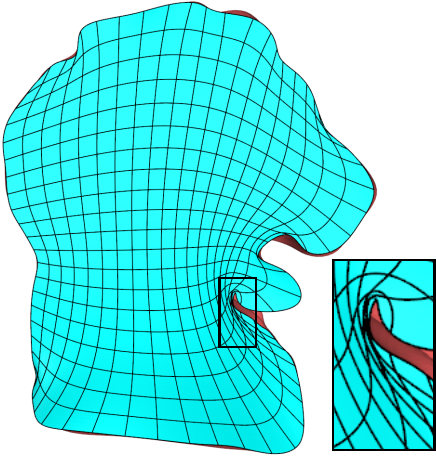}
  \centerline{}
\end{minipage}
\begin{minipage}{0.18\linewidth}
 \includegraphics[width=2.8cm]{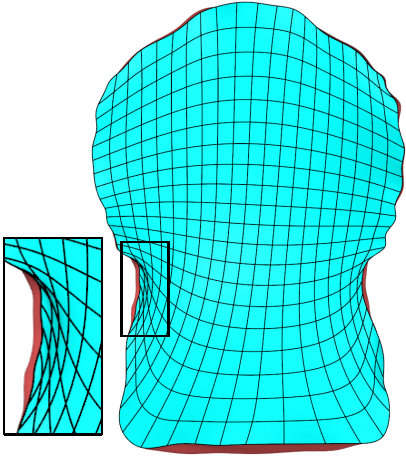}
  \centerline{}
\end{minipage}
\begin{minipage}{0.18\linewidth}
 \includegraphics[width=3.0cm]{blankfig}
  \centerline{}
\end{minipage}
\begin{minipage}{0.20\linewidth}
 \includegraphics[width=3.0cm]{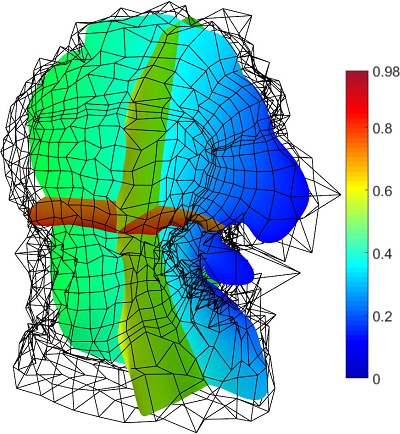}
  \centerline{}
\end{minipage}
\begin{minipage}{0.20\linewidth}
 \includegraphics[width=3.0cm]{blankfig}
  \centerline{}
\end{minipage}
}
\subfigure[Ours]{
\label{LionParameterization:Ours}
\begin{minipage}{0.18\linewidth}
 \includegraphics[width=3.0cm]{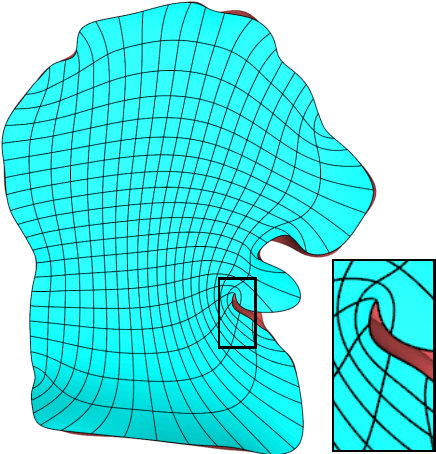}
  \centerline{}
\end{minipage}
\begin{minipage}{0.18\linewidth}
 \includegraphics[width=2.8cm]{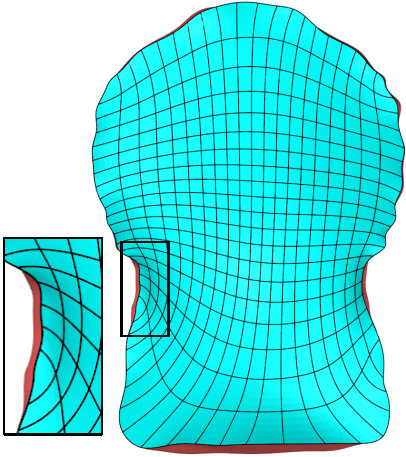}
  \centerline{}
\end{minipage}
\begin{minipage}{0.18\linewidth}
 \includegraphics[width=3.0cm]{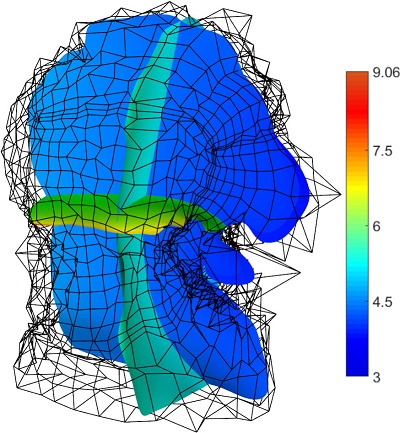}
  \centerline{}
\end{minipage}
\begin{minipage}{0.20\linewidth}
 \includegraphics[width=3.0cm]{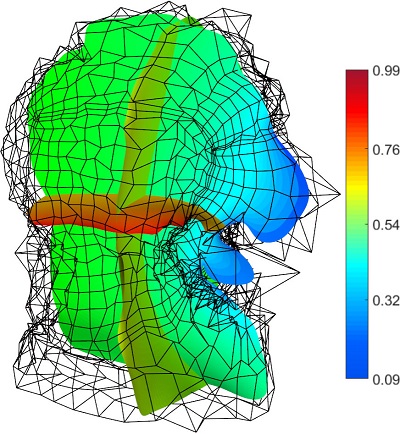}
  \centerline{}
\end{minipage}
\begin{minipage}{0.20\linewidth}
 \includegraphics[width=3.4cm]{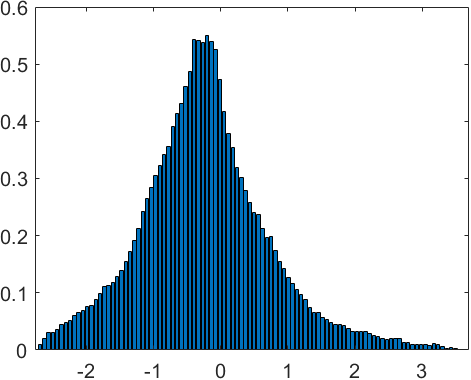}
  \centerline{}
\end{minipage}
}
    \caption{Volumetric parameterization of the lion model by~\subref{LionParameterization:Wangs} Wang's method and~\subref{LionParameterization:Ours} our method. The first two columns show the interior iso-parametric surfaces of the parameterization, the next two columns show the colormaps of $\kappa(J_{\mathbf{G}})$ and $\mathbf{G}_{\text{orth}}$ respectively, and the last column depicts the distributions of volume distortion $\log_2 (D_{vol}(\mathbf{G}))$. Here we omit the colormaps of $\kappa(J_{\mathbf{G}})$ and $\log_2 (D_{vol}(\mathbf{G}))$ for Wang's method due to its invalid parameterization.}
    \label{LionParameterization}
\end{figure}

\begin{figure}[!htbp]
    \centering
    \subfigure[]{
        \label{FemurParameterization:Interior}
        \includegraphics[width=0.40\textwidth]{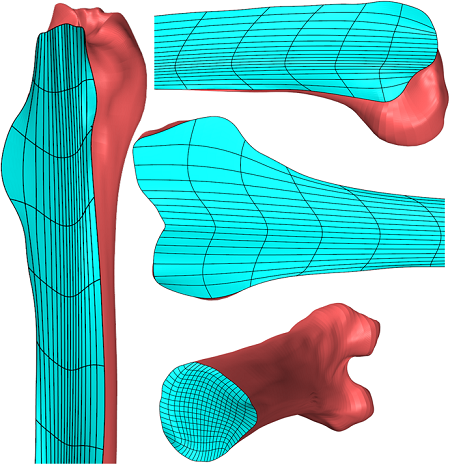}
    }
    \hspace{-0.03\textwidth}
    \subfigure[]{
        \label{FemurParameterization:ColorMap}
        \includegraphics[width=0.25\textwidth]{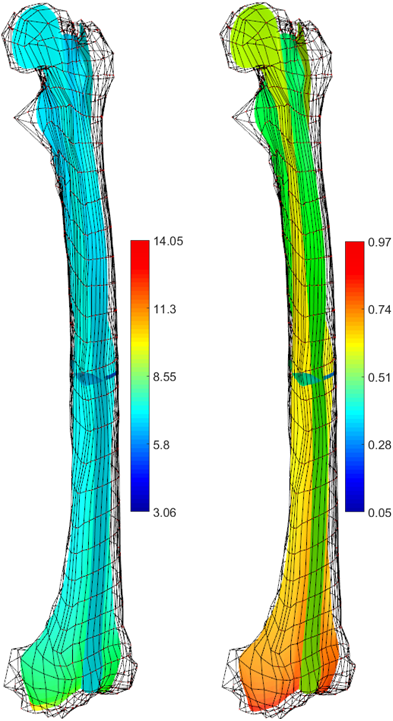}
    }
    \hspace{0.00\textwidth}
    \subfigure[]{
        \label{FemurParameterization:VolumeDistortion}
        \includegraphics[width=0.22\textwidth]{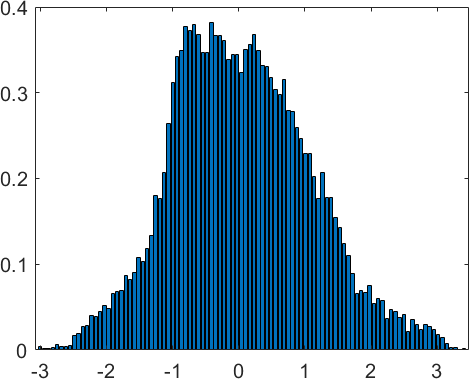}
    }
	\caption{Volumetric parameterization of the femur model by our method:~\subref{FemurParameterization:Interior} the interior iso-parametric surfaces of the parameterization,~\subref{FemurParameterization:ColorMap} the colormaps of $\kappa(J_{\mathbf{G}})$ and $\mathbf{G}_{\text{orth}}$,~\subref{FemurParameterization:VolumeDistortion} the distribution of volume distortion $\log_2 (D_{vol}(\mathbf{G}))$.}
\label{FemurParameterization}
\end{figure}

\begin{figure}[!htbp]
    \centering
    \subfigure[]{
        \label{MaxPlanckParameterization:InteriorOne}
        \includegraphics[width=0.16\textwidth]{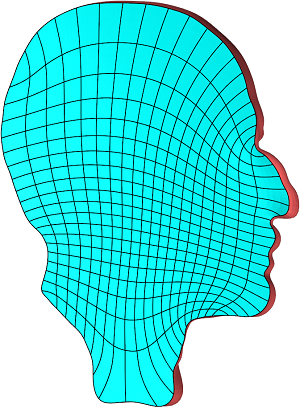}
    }
    \hspace{0.12\textwidth}
    \subfigure[]{
        \label{MaxPlanckParameterization:InteriorTwo}
        \includegraphics[width=0.135\textwidth]{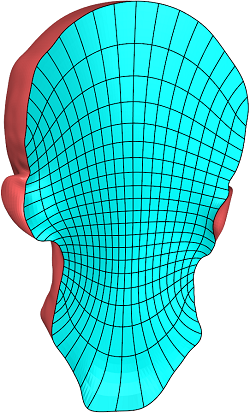}
    }
    \hspace{0.12\textwidth}
    \subfigure[]{
        \label{MaxPlanckParameterization:InteriorThree}
        \includegraphics[width=0.16\textwidth]{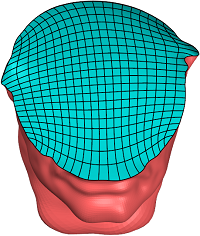}
    }
    \hspace{0.10\textwidth}
    \subfigure[]{
        \label{MaxPlanckParameterization:ConformalDistorColorMap}
        \includegraphics[width=0.20\textwidth]{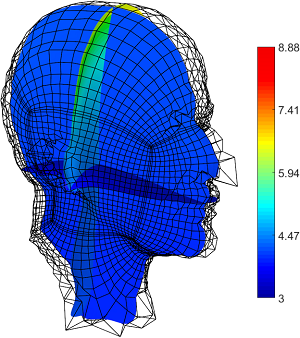}
    }
    \hspace{0.07\textwidth}
    \subfigure[]{
        \label{MaxPlanckParameterization:OrthDistorColorMap}
        \includegraphics[width=0.20\textwidth]{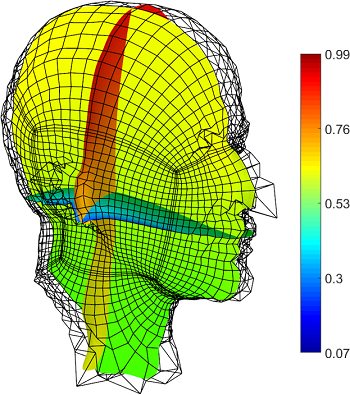}
    }
    \hspace{0.07\textwidth}
    \subfigure[]{
        \label{MaxPlanckParameterization:VolumeDistortion}
        \includegraphics[width=0.22\textwidth]{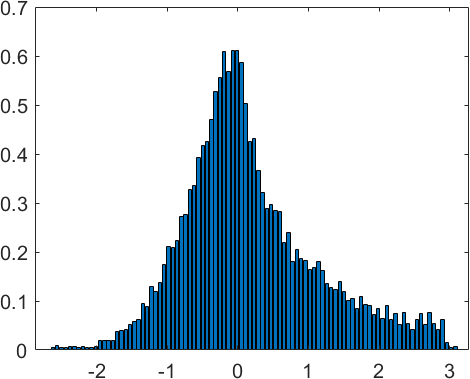}
    }
	\caption{Volumetric parameterization of the Max Planck model by our method:~\subref{MaxPlanckParameterization:InteriorOne} to~\subref{MaxPlanckParameterization:InteriorThree} show the interior iso-parametric surfaces of the parameterization,~\subref{MaxPlanckParameterization:ConformalDistorColorMap} and~\subref{MaxPlanckParameterization:OrthDistorColorMap} show the colormaps of $\kappa(J_{\mathbf{G}})$ and $\mathbf{G}_{\text{orth}}$ respectively and~\subref{MaxPlanckParameterization:VolumeDistortion} depicts the distribution of volume distortion $\log_2 (D_{vol}(\mathbf{G}))$.}
    \label{MaxPlanckParameterization}
\end{figure}

Table~\ref{PerformanceStatistics} summarizes the quantitative results (including distortion, orthogonality and running time) of the examples presented in Fig.~\ref{ToothParameterization}-\ref{MaxPlanckParameterization}, which also validate that our parameterization method is bijective, achieves lower distortion, better orthogonality and comparable computational efficiency than the other two state-of-the-art methods.

\begin{table*}[!htbp]
	\begin{center}
		\scriptsize
        \centering
        \renewcommand{\arraystretch}{1.5}		
        \begin{tabular}{|m{1.40cm}<{\centering}|c|c|c|m{0.85cm}<{\centering}|}\hline
        Model&$l$, $m$, $n$, $p$, $q$, $r$&Method&$\max(\kappa(J_{\mathbf{G}}))$, $\min(\mathbf{G}_{\text{orth}})$, $\max(\mathbf{G}_{\text{orth}})$, $\min(D_{vol}(\mathbf{G}))$, $\max(D_{vol}(\mathbf{G}))$&Time (m)\\
        \hline
		\multirow{3}{*}{\makecell{Tooth\\Fig.~\ref{ToothParameterization}}}&\multirow{3}{*}{20, 20, 7, 3, 3, 3}&Xu's&63.61, 0.01, 0.90, 0.04, 6.06&8.43\\
          \cline{3-5}
          &&Wang's&46.77, 0.01, 0.93, 0.19, 6.96&7.04\\
          \cline{3-5}
          &&Ours&9.50, 0.10, 0.99, 0.24, 5.66&3.31\\
          \hline
          \multirow{3}{*}{\makecell{Duck\\Fig.~\ref{DuckParameterization}}}&\multirow{3}{*}{17, 7, 11, 2, 2, 2}&Xu's&$\infty$, 0.00, 0.97, -4.36, 21.58&10.25\\
          \cline{3-5}
          &&Wang's&56.37, 0.00, 0.98, 0.04, 6.98&5.18\\
          \cline{3-5}
          &&Ours&10.94, 0.01, 0.99, 0.17, 7.15&6.62\\
          \hline
          \multirow{3}{*}{\makecell{Lion\\Fig.~\ref{LionParameterization}}}&\multirow{3}{*}{17, 17, 17, 3, 3, 3}&Xu's&---&---\\
          \cline{3-5}
          &&Wang's&$\infty$, 0.00, 0.98, -0.52, 10.85&24.51\\
          \cline{3-5}
          &&Ours&9.06, 0.09, 0.99, 0.14, 11.07&16.35\\
          \hline
          \multirow{3}{*}{\makecell{Femur\\Fig.~\ref{FemurParameterization}}}&\multirow{3}{*}{29, 13, 9, 2, 2, 2}&Xu's&---&---\\
          \cline{3-5}
          &&Wang's&$\infty$, 0.02, 0.97, -1.96, 11.03&15.88\\
          \cline{3-5}
          &&Ours&14.05, 0.05, 0.97, 0.12, 10.18&23.14\\
          \hline
          \multirow{3}{*}{\makecell{Max Planck\\Fig.~\ref{MaxPlanckParameterization}}}&\multirow{3}{*}{24, 24, 24, 3, 3, 3}&Xu's&---&---\\
          \cline{3-5}
          &&Wang's&17.53, 0.03, 0.98, 0.09, 10.36&39.79\\
          \cline{3-5}
          &&Ours&8.88, 0.07, 0.99, 0.14, 8.29&21.76\\
          \hline
		\end{tabular}
		\caption{Quantitative data for five models. Xu's method fails to run the last three data sets.}
		\label{PerformanceStatistics}	
	\end{center}
\end{table*}

\section{Conclusions and future work}
\label{sec:conclusion}
Volumetric parameterization of computational domains is an essential step in IGA as mesh generation in finite element analysis. In this work, we propose a novel volumetric parameterization approach which includes three main steps: computing an initial harmonic mapping, constructing a bijective mapping by solving a max-min constrained optimization problem and improving parameterization quality using MIPS. Coarse-to-fine and divide-conquer strategies are applied to efficiently solve the optimization problems. Experimental examples demonstrate that our approach can produce a bijective and low-distortion volumetric parameterization which outperforms other state-of-the-art methods.

Regarding the future work, there are two major problems which are worthy of further investigation. Firstly, the running time reported in Table~\ref{PerformanceStatistics} indicates that our algorithm is still computationally expensive, which may be accelerated by GPU computation. Secondly, currently our approach only deals with geometries of genus-zero. For geometries with high genus and more complex boundaries, the domain decomposition method may have to be applied to partition the domain into simply connected regions and then each region can be parameterized using the technique proposed in this paper.

\section*{Acknowledgement}
\label{sec:acknowledgement}
This work was supported by the NSF of China (No. 11571338, 61877056), China Postdoctoral Science Foundation (No. 2018M632548), the Fundamental Research Funds for the Central Universities (No. WK0010460007), and the Open Project Program of the State Key Lab of CAD$\&$CG (No. A1819), Zhejiang University.

\end{document}